\documentclass[11pt,authoryear]{elsarticle}
\usepackage[left=2.5cm, right=2.5cm, top=3.5cm, bottom=3cm, footskip=0.5cm]{geometry}
\usepackage{amsmath, amsfonts, amsthm, amssymb}
\usepackage{verbatim}
\usepackage{hyperref}
\usepackage{color}
\usepackage{graphicx}
\usepackage{enumerate}
\usepackage{wrapfig}
\usepackage{eurosym}
\usepackage{subcaption}
\usepackage{ulem}
\usepackage{cases} 
\usepackage{fancyhdr}
\usepackage{multirow}
\usepackage{mathrsfs}

\numberwithin{equation}{section}
\newtheorem{theorem}{Theorem}[section]
\newtheorem{corollary}[theorem]{Corollary}

\theoremstyle{definition}
\newtheorem{definition}[theorem]{Definition}
\newtheorem{remark}[theorem]{Remark}
\newtheorem{assumption}[theorem]{Assumption}
\newtheorem{sassumption}[theorem]{Standing Assumption}

\numberwithin{equation}{section}
\usepackage[none]{hyphenat}

\newcommand{\RR}{\mathbb{R}}

\newcommand{\Ir}{\mathbf{I}}
\newcommand{\NN}{\mathbb{N}}
\newcommand{\Ff}{\mathcal{F}}

\newcommand{\Ttt}{\mathcal{T}}
\newcommand{\ft}{\mathfrak{t}}

\newcommand{\cN}{\mathcal{N}}
\newcommand{\cLN}{\mathcal{LN}}
\newcommand{\cW}{\mathcal{W}}

\newcommand{\cE}{\mathcal{E}}

\newcommand{\cZ}{\mathcal{Z}}

\newcommand{\cC}{\mathcal{C}}
\newcommand{\cG}{\mathcal{G}}
\newcommand{\cH}{\mathcal{H}}
\newcommand{\cA}{\mathcal{A}}

\newcommand{\cL}{\mathcal{L}}

\newcommand{\cU}{\mathcal{U}}
\newcommand{\EE}{\mathbb{E}}
\newcommand{\dr}{\mathrm{d}}
\newcommand{\OneN}{\{1,\hdots,N\}}

\newcommand{\VV}{\mathbb{V}}
\newcommand{\bOne}{\mathbf{1}}

\newcommand{\pf}{\mathfrak{p}}
\newcommand{\ff}{\mathfrak{f}}
\newcommand{\cc}{\mathfrak{c}}

\newcommand{\Vv}{{V}}

\newcommand{\titre}{Modeling the impact of Climate transition on real estate prices}

\usepackage{fancyhdr}
\begin{document}

\begin{frontmatter}

\title{\titre}
\date{\today}

%\author[4]{G\'{e}raldine Bouveret}
%\author[1]{Jean-Fran\c{c}ois Chassagneux}
%\author[3]{Smail Ibbou}
%\author[2,5]{Antoine Jacquier}
\author[1,2,3]{Lionel Sopgoui}

\address[1]{ Laboratoire de Probabilités, Statistique et Modélisation (LPSM), Université Paris Cité}
\address[2]{Department of Mathematics, Imperial College London}
\address[3]{Validation des modèles, Direction des risques, Groupe BPCE}
%\address[4]{Climate Risks Research Department, Rimm Sustainability Ltd}
%\address[5]{Alan Turing Institute}

\journal{arXiv}

\begin{abstract}
In this work, we propose a model to quantify the impact of the climate transition on a property in the housing market. We begin by noting that property is an asset in an economy. That economy is organized in sectors, driven by its productivity which is a multidimensional Ornstein-Uhlenbeck process, while the climate transition is declined thanks to the carbon price, a continuous deterministic process. We then extend the \textit{sales comparison approach} and the \textit{income approach} to evaluate an energy inefficient real estate asset. We obtain its value as the difference between the price of an equivalent efficient building following an exponential Ornstein-Uhlenbeck as well as the actualized renovation costs and the actualized sum of the future additional energy costs (before and after the renovation date). These costs are due to the inefficiency of the building, before an optimal renovation date which depends on the carbon price process. Moreover, since the renovation increases the efficiency of the building, which is random, the future additional energy costs become smaller and even zero if the optimal energy efficiency is reached. Finally, we carry out simulations based on the French economy and the house price index of France. The findings support the conclusion that the order of magnitude of the depreciation obtained by our model is the same as the empirical observations.

\end{abstract}


\end{frontmatter}

\footnotesize This research is part of the PhD thesis in Mathematical Finance of Lionel Sopgoui whose works are funded by a CIFRE grant from BPCE S.A. The opinions expressed in this research are those of the authors and are not intended to reflect the opinions or official positions of BPCE S.A.

\footnotesize We would like to thank Jean-Fran\c{c}ois Chassagneux, Antoine Jacquier, Smail Ibbou, and Géraldine Bouveret for helpful comments on an earlier version of this work.

%\section*{Updated abstract}

\normalsize
%\tableofcontents
%%%%%%%%%%%%%%%%%%%%%%%%%%%%%%%%%%%%%%%
\newpage
\paragraph{Notations}
\begin{itemize}
    \item $\NN$ is the set of non-negative integers, $\NN^{*} := \NN\setminus\{0\}$, and $\mathbb{Z}$ is the set of integers.
    \item $\RR^d$ denotes the $d$-dimensional Euclidean space, $\RR_{+}$ is the set of non-negative real numbers, $\RR_{+}^{*} := \RR_{+}\setminus\{0\}$.
    \item $\bOne := (1,\hdots,1) \in\RR^{I}$.
    \item $\RR^{n\times d}$ is the set of real-valued $n\times d$ matrices ($\RR^{n\times 1} = \RR^{n}$), $\Ir_n$ is the identity $n\times n$ matrix.
    \item $x^i$ denotes the $i$-th component of the vector $x \in \RR^d$. For all $A := (A^{ij})_{1\leq i,j\leq n}\in\RR^{n\times n}$, we denote by~$A^\top := (A^{ji})_{1\leq i,j\leq n}\in\RR^{n\times n}$ the transpose matrix, and $\lambda(A)$ denotes the spectrum of $A$.
    \item For all $x,y\in\RR^d$, we denote the scalar product $x^\top y$, the Euclidean norm~$ | x | := \sqrt{x^\top x}$ and for a matrix~$M\in\RR^{d\times d}$, we denote
\begin{equation*}
    | M |:= \sup_{a\in\RR^d, |a| \leq 1}   |Ma|  \label{ct-eq:norm}.
\end{equation*}
\item $(\Omega, \mathcal{H}, \mathbb{P})$ is a complete probability space.

\item  For $p \in [1,\infty]$, ${E}$ is a finite dimensional Euclidian vector space and for a $\sigma$-field $\cH$, $\cL^p(\cH,{E})$, denotes the set of $\cH$-measurable random variable $X$ with values in ${E}$ such that $\Vert X \Vert_{p} := \left(\EE\left[ |X|^p\right] \right)^{\frac1p}<\infty$ for $p < \infty$ and for $p = \infty$, $\Vert X \Vert_{\infty} := \mathrm{esssup} |X(\omega)| < \infty$. 
\item For a filtration $\mathbb{G}$, $p \in [1,+\infty]$ and $I\in \NN^*$, $\mathscr{L}^p_{+}(\mathbb{G},(0,\infty)^I)$ is the set of continuous-time processes that are $\mathbb{G}$-adapted valued in $(0,\infty)^I$ and  which satisfy
\begin{equation*}
    \lVert X_t \rVert_p < \infty \text{ for all } t\in\RR_+.
\end{equation*}
\item If $X$ and $Y$ are two random variables $\RR^d$-valued, for $x\in\RR^d$, we note $Y|X=x$ the conditional distribution of $Y$ given $X=x$, and $Y|\Ff$ the conditional distribution of 
$Y$ given the filtration~$\Ff$.
\item If $f: \RR\to\RR, t\mapsto f(t)$ is a differentiable function, we note $\dot f$ its first derivative.
\end{itemize}

\section*{Introduction}

Real estate is a large part of capital stock and a large component of economic wealth. For example, Figure~\ref{ct-fig:housing_share} from~\cite{oecd2024realestate} shows that dwellings represent between 10\% and 40\% of annual investment and between 15\% and 25\% of annual households disposable income. This explains the need to price real estate assets. Determining the market value of a building may be essential for accounting needs (portfolio construction or management, asset price evolution, etc.), for making decisions (build, buy, or sell a property), or for applying for loans (collateral). In his dissertation, \cite{schulz2003valuation} presents three approaches for real estate valuation. The \textit{sales comparison approach} consists in using the transaction prices of highly comparable and recently sold properties to determine the market value of a given building. The \textit{income approach} where the value of a property is the discounted sum of the (imputed) rent less all operating costs over its residual lifetime.
The \textit{cost approach} which consists in assuming that the value of a building is equal to
the amount it would cost to buy the land (with identical characteristics) and to build a new building (with the same characteristics) on it. 

The literature suggests several models for each approach. A real estate market model is proposed by~\cite{fabozzi2012pricing} to price real estate derivatives and then used for the calculation of the LGD by~\cite{frontczak2015modeling}. These works model the price of a property as an exponential Ornstein-Uhlenbeck process while \cite{moodys2022} use a geometric Brownian motion. Laspeyres-Paasche housing index (see~\cite{gravelle2004microeconomics}) is estimated based on the transactions price of buildings sold in the previous period, therefore, valuate a property using the housing price index is an example \textit{sales comparison approach}. For the \textit{income approach}, the expected net (imputed) income is the tricky quantity to determine. \cite{schulz2003valuation} takes into account the relative maintenance and repair expenditures, the depreciating rate, the real estate tax rate, the real interest rate, the inflation of the rent (or imputed rent -- the rental price an individual would pay for an asset they own --), etc. These two approaches are used in this work.

Climate change has a deep impact on human societies and their environments. One of the components of climate risk is transition risk which relates to the potential economic and
financial losses associated with the process of adjusting towards a low-carbon economy. Transition risk is becoming increasingly
important in all parts of the economy, in particular in real estate. Residential or commercial buildings are one of the biggest greenhouse gases (GHG) emitters. Using Eurostat GHG emissions data (see~\cite{eurostat2022}), we find in Figure~\ref{ct-fig:part_GHG_emissions_on_households} that, between 2008 and 2021, around 45\% of the total household emissions came from residential heating and cooling. And as shown in Figure~\ref{ct-fig:part_GHG_emissions_on_total}, they represent around 9\% of the aggregate emissions, in the EU economic region. This is why renovation of buildings constitutes a central challenge in climate policies. The Energy Performance of Buildings Directive (EPBD) from~\cite{directive200291eu} and~\cite{eu2022epbd}, introduced in 2002 by the European Commission and revised in 2010 and later, is a key instrument to increase the energy performance of buildings across the EU. Similarly to the carbon price, it is a way to implement climate transition. It consists in ranking buildings in terms of their energy efficiency (EE)\footnote{in ton of CO$_2$ emissions per square meter per year or kilowatt hour per square meter per year} by using letters from A to G (where A is the most efficient while G is the less). \cite{aydin2020capitalization} found that EE is capitalized quite precisely into home prices in the Dutch housing market. \cite{de2016price} found in a study on 1507 homes in Spain that dweelings labelled A, B or C are valued at between 5.4\% and 9.8\% higher price compared to D, E, F or G rated home. \cite{franke2019energy} also highlight that, in the rental decision-making, EE achieves a high importance score similar to that of rent, price and location respectively.

 There is a great deal of statistical work on the effect of climate change on real estate, but very little modeling. This last point interests us here. To introduce the climate transition, we draw inspiration from~\cite{ter2021german} who write the price difference per square meter between two properties with different energy efficiency as the sum of the discounted value of (expected) energy cost differences. We will enhance their work by additionally considering the renovation costs to improve energy efficiency and the optimal renovation date which will depend on the trajectory of energy prices (written as a function of the carbon price). Initially, the owner of the building incurs additional energy costs as a result of the inefficiency of their property. Then, they may decide to spend money on renovations to improve their building energy efficiency. After the renovation, either the energy efficiency level is reached so that they no longer incur additional energy costs. Or it is not reached, and they will continue suffering energy costs, but smaller than before.

The remainder of the present work is organized as follows. Because the real estate industry is a sector of the economy, in Section~\ref{ct-sec:cont time}, we will describe the economy concerned. We model, in Subsection~\ref{ct-sec:col property}, a \textit{building} using both an exponential Ornstein-Uhlenbeck (O.U.) and a discounted sum of an income over its residual lifetime. Section~\ref{sec:estim calib} is dedicated to estimations and simulations, while Section~\ref{sec:discussion} focuses on discussion.

\section{Standing assumptions}\label{ct-sec:cont time}

We assume that real estate assets are part of an economy divided into sectors $I\in\NN^*$, driven by dynamic and
stochastic productivity, and subject to climate transition modeled by a dynamic and
deterministic carbon price. Inspired by \cite{bouveret2023propagation}[Section 2] which models, in discrete time, a multisectoral closed economy subject to the carbon price. We introduce the following standing assumption which describes the productivity, which is considered to have stationary Ornstein-Uhlenbeck dynamics.
\begin{sassumption}\label{ct-sassump:OU}
    We define the $\RR^I$-valued process~$\mathcal{A}$ which evolves  according to 
\begin{equation}\label{ct-eq:VAR}
     \left\{
     \begin{array}{rl}
      \dr\cZ_t &= -\Gamma\cZ_t \dr t + \Sigma \dr B_t^{\cZ}\\
     % (\mu + \varsigma \cZ_t) &= \mu + \varsigma \cZ_t\\
     \dr \cA_t &= \left(\mu + \varsigma \cZ_t\right) \dr t
     \end{array}\quad\textrm{for all } t\in\RR_+,
     \right.
     \end{equation}
     where $(B_t^{\cZ})_{t\in\RR^*}$ is a $I$-dimensional Brownian motion, and where the constants $\mu, \mathcal{A}_0 \in \RR^I$, the matrices~$\Gamma,\Sigma \in\RR^{I\times I}$, $\cZ_0 \sim \cN\left(0,  \Sigma \Sigma^\top \right)$, and $0 < \varsigma \le 1$ is a fixed parameter controlling noise intensity: it will be used later to obtain a tractable proxy of the firm value. Moreover, $\Sigma$ is a positive definite matrix and $-\Gamma$ is a Hurwitz matrix i.e. its eigenvalues have strictly negative real parts.
\end{sassumption}

We also introduce the following filtration $\mathbb{G}:=(\mathcal{G}_t)_{t\in\RR^*}$ with $\cG_0 := \sigma(\cZ_0)$ and for $t>0$, $\mathcal{G}_t := \sigma\left(\left\{\cZ_0, B_s^{\cZ}: s\leq t\right\}\right)$. 

\begin{remark}[O.U. process]\label{ct-rem:VAR1}
We have the following results on O.U. that we will use later on:
    \begin{enumerate} 
        \item According to~\cite{gobet2016perturbation}[Proposition 1], if one assumes that $\cZ_0$ and $B^{\cZ}$ are independent and $\cZ_0$ is square integrable, then, there exists a unique square integrable solution to the $I$-dimentional Ornstein-Uhlenbeck process~$\cZ$ satisfying $\dr\cZ_t = -\Gamma\cZ_t \dr t + \Sigma \dr B_t^{\cZ}$, represented as
        \begin{equation*}
            \cZ_t = e^{-\Gamma t} \left(\cZ_0 + \int_{0}^{t} e^{\Gamma u} \Sigma \dr B_u^{\cZ} \right),\quad\textrm{for all } t\in\RR_+.
        \end{equation*}
         Additionally, for any $t,h \geq 0$, the distribution of $\cZ_{t+h}$ conditional on $\cG_t$ is Gaussian~$\cN\left(M^{\cZ,h}_{t}, \Sigma^{\cZ,h}_{t}\right)$, with the mean vector
        \begin{equation}
            M^{\cZ,h}_{t} := \EE[\cZ_{t+h}|\cG_t] = e^{-\Gamma h} \cZ_t,
        \end{equation}
     and the covariance matrix
     \begin{equation}
         \Sigma^{\cZ,h}_{t} := \VV[\cZ_{t+h}|\cG_t] = \int_{0}^{h} e^{-\Gamma u} \Sigma \Sigma^\top e^{-\Gamma^\top u} \dr u.
     \end{equation}
     \item Since $-\Gamma$ is a Hurwitz matrix, then if we note $\lambda_\Gamma :=\max_{\lambda\in\lambda(\Gamma)} Re(\lambda)$, there exists $c_\Gamma>0$ so that $\lVert e^{-\Gamma t}\rVert < c_\Gamma e^{-\lambda_\Gamma t}$ for all~$t\geq 0$. Therefore, according to~\cite{gobet2016perturbation}[Proposition 2], $\cZ$ has a unique stationary distribution which is Gaussian with mean $0$ and covariance $\int_{0}^{+\infty} e^{-\Gamma u} \Sigma \Sigma^\top e^{-\Gamma^\top u}\dr u$. 
    \item We can show in~\ref{ct-app:OU} that for any $t, h \geq 0$, we have
        \begin{equation*}
            \cA_{t+h} = \cA_{t} + \int_{t}^{t+h} (\mu + \varsigma \cZ_s) \dr s = \mu h + \varsigma \int_{t}^{t+h} \cZ_s \dr s,
        \end{equation*}
        and conditionally on $\cG_t$, $\cA_{t+h}$ has an $I$-dimensional normal distribution with the mean vector
        \begin{equation}
            M^{\cA,h}_{t} :=  \mu h + \varsigma\Upsilon_{h}\cZ_t +  \cA_t,
        \end{equation}
    with \begin{equation}\label{ct-eq:Upsilon}
    \Upsilon_{h} := \int_{0}^{h} e^{-\Gamma s} \dr s = \Gamma^{-1}(\Ir_I-e^{-\Gamma h}),
\end{equation}
     and the covariance matrix
     \begin{equation}\label{ct-eq:Ma_ht}
         \Sigma^{\cA,h}_{t} := \varsigma^2 \Gamma^{-1} \left(\int_{0}^{h} \left(e^{-\Gamma u} - \Ir_I \right) \Sigma\Sigma^\top \left(e^{-\Gamma u} - \Ir_I \right) \dr u \right) (\Gamma^{-1})^\top = \varsigma^2\int_{0}^{h} \Upsilon_{u} \Sigma\Sigma^\top \Upsilon_{u}^\top \dr u .
     \end{equation}
    \item For later use, we define 
        \begin{align}\label{ct-de A circ}
            \cA^\circ_t := \mathcal{A}_t - \mathcal{A}_0,
        \end{align}
        and observe that $(\cA^\circ_t,\cZ_t)_{t \ge 0}$ is a Markov process.
    \end{enumerate}
\end{remark}

For the whole economy, we introduce a deterministic and exogenous carbon price in euro/dollar per ton.
It allows us to model the impact of the transition pathways on the whole economy. We will note $\delta$ the complete carbon price process. We shall then assume the following setting. 
\begin{sassumption}\label{ct-sassc:price}
We introduce the carbon price process.

Let $0 \le t_\circ < t_\star$ be given. 
The sequence $\delta$ satisfies
\begin{itemize}
    \item for $t \in [0;t_\circ]$, $\delta_t = \delta_0\in (\RR_+)^I$, namely the carbon price is constant;
    \item for $t \in (t_\circ,t_\star)$, $\delta_t \in (\RR_+)^I$, the carbon price may evolve;
    \item for $t \ge t_\star$, $\delta_t = \delta_{t_\star} \in (\RR_+)^I$, namely the carbon price is constant.
\end{itemize}
\noindent We assume moreover that $t\mapsto \delta_t$ is $\mathcal{C}^1(\RR_+, \RR_+)$.
\end{sassumption}

\paragraph{An example of carbon price process}
We assume the regulator fixes $t_\circ\geq 0$ when the transition starts and the transition horizon time~$t_\star>t_\circ$, the carbon price at the beginning of the transition~$P_{carbon} > 0$, at the end of the transition~$\delta_{t_\star} > P_{carbon}$, and the annual growth rate~$\eta_\delta > 0$. 
Then, for all~$t\geq 0$,
 \begin{equation}\label{ct-eq:carbon price}
    \delta_{t}= \left\{
    \begin{array}{ll}
    \displaystyle P_{carbon}, & \mbox{if } t \leq t_\circ,\\
    \displaystyle P_{carbon} e^{\eta_\delta(t-t_\circ)}, & \mbox{if } t\in(t_\circ,t_\star],\\
    \displaystyle \delta_{t_\star} = P_{carbon} e^{\eta_\delta(t_\star-t_\circ)}, &  \mbox{otherwise}.
    \end{array}
\right.
\end{equation} 
 In the example above that will be used in the rest of this work, we assume that the carbon price increases. However, there are several scenarios that could be considered, including a carbon price that would increase until a certain year before leveling off or even decreasing. We also assume an unique carbon price for the entire economy whereas we could proceed differently. For example, to prevent social upheaval, carbon pricing could be adjusted -- intensified for industrial production while remaining stable or even eliminated for households. The framework can be adapted to various sectors as well as scenarios.\\

\section{Valuation of a property under climate transition}\label{ct-sec:col property}
 The problem here is to model the real estate market in the presence of the climate transition risk. We will use two of the three approaches mentioned in the introduction: the \textit{sales comparison approach} and the \textit{income approach}. For efficient buildings, we use the first approach. Precisely, we will write the price of a property as the product of its surface area, its initial price, and the house price index (the latter described by an exponential Ornstein–Uhlenbeck dynamics). For the inefficient ones, we adopt the second approach. Explicitly, in addition to all the costs involved in owning the property, we will consider the energy costs due to energy inefficiency and the potential renovation costs.\\

According to~\cite{ter2021german}, in discrete time, the price difference per square meter between two properties, one of which with the highest EE label as the reference point (A+), should be solely explained by the sum of the discounted value of (expected) energy cost (noted~$EC$) differences:
\begin{equation}\label{ct-eq:tergerman}
    P^j_t-P^{A+}_t = -\sum_{h=1}^{T} \frac{EC_{t+h}^j-EC^{A+}_{t+h}}{(1+r)^h}.
\end{equation}
This equation takes the perspective of a potential buyer who weighs the
options between buying the efficient property with score~$A+$ at a higher price and enjoying the lower energy costs, and buying the inefficient one with score~$j$ at a discount that reflects the expected increased energy costs at time $T$. Moreover, energy costs $EC^j_t$ are the simple product of the expected energy price and a constant factor measuring the EE.\\

We extend the~\cite{ter2021german} work here by assuming the following.
\begin{enumerate}
    \item In the absence of climate transition, the housing price follows an exponential Ornstein-Uhlenbeck process.
    \item Each dwelling consumes a given quantity of energy per square meter, which is used to determine its energy efficiency, noted $\alpha$ and expressed in kilowatts per square meter (KWh/m2).
    \item As a consequence, the dwelling price is depreciated (or appreciated) by the actualized sum of future energy costs. 
    \item Once a certain level of energy efficiency $\bar\alpha$ is reached, the market is insensitive to this factor. 
    \item The energy price is a deterministic function~$\ff$ of two variables, the first variable is the carbon price and the second is the source of energy.
    \item During the life of the property, if $\alpha > \bar\alpha$ (i.e. it is not efficient), the owner may undertake renovations which move the energy efficiency from~$\alpha$ to~$\alpha^\star$, and whose cost per square meter is a function~$\cc$ of its energy efficiencies.
    \item The density of the probability that the owner spends $\cc(\alpha, \alpha^\star)$, to move the efficiency of its property from $\alpha$ to $\alpha^\star$ is $g: \RR_+^2 \to \RR_+$.
    \item The date of renovation~$\ft$ of a dwelling is unknown, but is to be optimized.
    \item After renovations, the price of the building becomes insensitive to energy costs.
\end{enumerate}

\begin{assumption}[Housing price without climate transition]\label{ass:housing market} We consider here two ways to model a housing price:
    \begin{enumerate}
        \item The market value of the building at $t\geq 0$, is given by
    \begin{align}\label{ct-eq:housing EOU}
        C_{t} &:= R C_{0} e^{K_t} ,
    \end{align}
    where
\begin{subequations}
\begin{align}
\displaystyle\dr K_t &= \left(\dot{\chi}_t+ \nu(\chi_t - K_t) \right) \dr t + \overline{\sigma} \dr\overline{B}_t, \label{ct-eq:coldynamics2}\\
\dr \overline{B}_t &= \rho^\top \dr B^{\cZ}_t + \sqrt{1-\lVert\rho\rVert^2} \dr\overline\cW_t,
\label{ct-eq:coldynamics3}
\end{align}
\end{subequations}
where $(\overline\cW_t)_{t\in\RR_+}$ is a standard Brownian motion independent to~$B^{\cZ}$ introduced in Standing Assumption~\ref{ct-sassump:OU} and driving the productivity of the economy. Moreover, $C_{0}$, $r$, $R, \overline{\sigma} > 0$, $\rho\in\RR_+^I$, and $\chi\in\mathcal{C}^1(\RR_+, \RR_+)$. 
    \noindent We introduce the following filtration $\mathbb{U}:=(\cU_t)_{t\in\RR^*}$ with for $t\geq 0$, $\cU_t := \sigma\left(\left\{ \overline{W}_s, B^{\cZ}_s: s\leq t\right\}\right)$.  
\item Owning a building allows to generate an ( the (imputed) income cash-flow process~$(D_t)_{t\geq 0}$ which is continuous and $\mathbb{U}$-adapted, therefore for all~$t\geq 0$, an another way to write the price $C_{t}$ of the building is
\begin{align}\label{ct-eq:housing DCF}
    C_{t} &= R\EE\left[\int_{0}^{+\infty} e^{- \bar r s} D_{t+s}\dr s \middle|\cU_t\right],
\end{align}
with $\bar r>0$.
    \end{enumerate}
\end{assumption}
The previous assumption yields the following observations:
\begin{itemize}
    \item $C_{0}$ in~\eqref{ct-eq:housing EOU}, is the value per square meter of the building at time~$0$ (in euros/m$^2$ for example) while $K$ is the log of the housing price index whose dynamics is~\eqref{ct-eq:coldynamics2} inspired by~\cite{frontczak2015modeling, fabozzi2012pricing}. It characterizes the returns of the housing market which are correlated and fluctuates over time around a long-term average~$\chi$.
    \begin{figure}[!ht]
    \centering
        \includegraphics[width=0.88\textwidth]{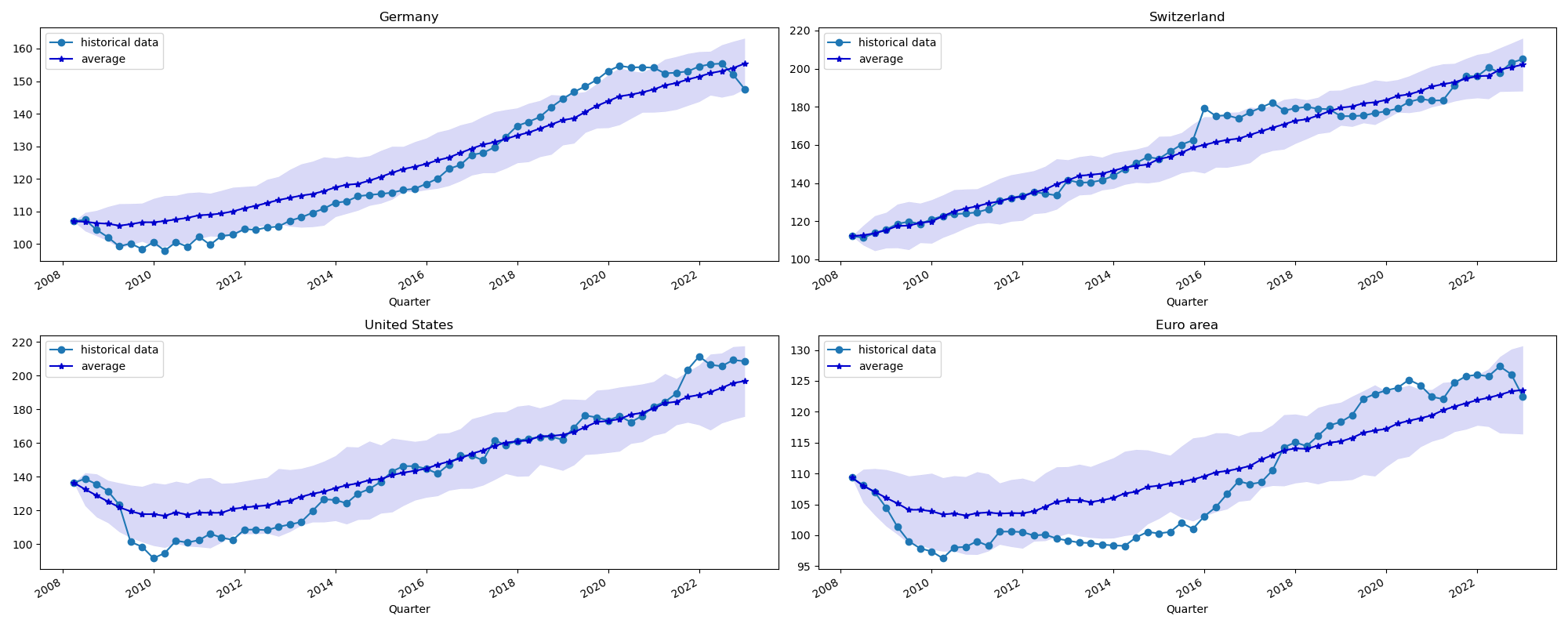}
        \caption{Mean and 95\% confidence interval of the housing price index of four countries}
        \label{ct-fig:housing_price}
    \end{figure}
    In Figure~\ref{ct-fig:housing_price}, we visualize historical housing price index of four countries calibrated on~\eqref{ct-eq:coldynamics2} with the average long-term level~$\chi$ linear in time. This means that~$\chi_t\varrho t + \vartheta$, for all~$t\in\RR_+$, is a good choice.
    \item Equation~\eqref{ct-eq:housing DCF} refers to the fact that owning a building leads to additional income such as depreciation, maintenance costs, opportunity costs of capital (rent received or saved), or flow of various taxes (see~\cite{schulz2003valuation}[Chapter 2]). Precisely, if we denote $P_s := \EE[D_{s}|\mathcal{U}_t]$, for~$s\geq t$, we could write, 
    \begin{align}
        \dot{P_s} = (\pi-\gamma-\upsilon-m) P_s,
    \end{align}
    where $m$ is related to maintenance and repair expenditures, $\gamma$ is related to the depreciating rate, $\upsilon$ is related the real estate tax rate, $\pi$ is related the (imputed) rent rate.
    \item Moreover, recall that $B^{\cZ}$ is the noise of the productivity process of the economy, the definition of $\overline{B}$ in \eqref{ct-eq:coldynamics3} allows then to link the real estate market and the productivity (to verify this, one could look at the supply and use tables of~\cite{insee2023sut}\footnote{The French National Institute of Statistics and Economic Studies} to see the links between real estate activities and other economic sectors).
\end{itemize}
\noindent We now turn our attention to the law of the solution of~\eqref{ct-eq:coldynamics2}. We have for $0\leq t \leq T$,
\begin{equation}
    \begin{split}
    K_{T} &= \chi_{T} - \left(\chi_0 - K_0 \right)e^{-\nu T} + \overline{\sigma}\int_{0}^{T} e^{-\nu(T-s)} \dr \overline{B}_s^{n}\\
    &= \chi_{T} -  \left(\chi_t - K_t\right)e^{-\nu (T-t)} + \overline{\sigma}\rho^\top \int_{t}^{T} e^{-\nu(T-s)} \dr B_s^{\cZ}  + \overline{\sigma}\sqrt{1-\lVert\rho\rVert^2}  \int_{t}^{T} e^{-\nu(T-s)} \dr \overline \cW_s.
    \end{split}
\end{equation}
We therefore get the distribution of $K_{T}$ conditional on $\cG_t$ (and not conditional on $K_t$) of is Gaussian, 
\begin{small}
\begin{align}\label{ct-eq: cond law Y_T}
        &\cN\left(\chi_{T} - \left(\chi_0-K_0 \right)e^{-\nu T} + \overline{\sigma}\rho \int_{0}^{t} e^{-\nu(T-s)} \dr B^{\cZ}_s\right.\\
        &\qquad\qquad\left.\frac{(\overline{\sigma}\lVert\rho\rVert)^2)}{2\nu}\left(1-e^{-2\nu (T-t)}\right) + \frac{(\overline{\sigma})^2 (1-\lVert\rho\rVert^2)}{2\nu}\left(1-e^{-2\nu T}\right)\right).
\end{align}
\end{small}

\noindent We can rewrite~\eqref{ct-eq:housing EOU} is, for $0\leq t$, 
\begin{equation*}
        C_{t} = R C_0 \exp{\left( \chi_{t}- \left(\chi_0 - K_0\right)e^{-\nu t} + \overline{\sigma}\int_{0}^{t} e^{-\nu(t-s)} \dr \overline{B}_s^{n}\right)}.
\end{equation*}

The following corollary gives the conditional distribution of the collateral. Its proof is straightforward and is directly derived from~\eqref{ct-eq: cond law Y_T}.
\begin{corollary}\label{ct-cor: cond law coll 2}
    For $0\leq t \leq T$, the law of
    $C_{t+T} = R C_0\exp{(K_{t+T})}$ conditional on $\cG_t$ is log-Normal $\cLN(m_{t, T}, v_{t, T})$ with
    \begin{equation}\label{ct-eq:mean col}
        m_{t, T} := \log{(R C_0)} + \chi_{t+T} - \left(\chi_0-K_0 \right)e^{-\nu (t+T)} + \overline{\sigma}\rho^\top \int_{0}^{t} e^{-\nu(t+T-s)} \dr B^{\cZ}_s,
    \end{equation}
    and
    \begin{equation}\label{ct-eq:var col}
        v_{t, T} :=  \frac{(\overline{\sigma}\lVert\rho\rVert)^2)}{2\nu}\left(1-e^{-2\nu T}\right) + \frac{(\overline{\sigma})^2 (1-\lVert\rho\rVert^2)}{2\nu}\left(1-e^{-2\nu (t+T)}\right).
    \end{equation}
\end{corollary}
We now turn our attention to the pricing of a building, taking into account its energy efficiency. We would like to calculate the value of a dwelling at time time~$t$. We use the actualized sum of the cash flows before the renovation date (taking into account the additional energy costs due to inefficiency of the building), at the renovation date, and after the renovation date (when the building becomes efficient). Moreover, the agent rationally chooses the date of renovation maximizing the value of his property. We have

\begin{definition}[Housing price with climate transition]\label{ct-ass:housing}
    The market value of the building serving as the collateral to firm~$n$ at $t\geq 0$, given the carbon price sequence~$\delta$, is represented by
\begin{align}\label{ct-eq:def house price}
    \cC_{t,\delta} := R\displaystyle\operatorname*{ess~sup}_{\theta\geq t} \EE\left[\begin{array}{l}
 \int_{t}^{\theta}\left[D_{u}-(\alpha-\alpha^\star) \ff(\delta_u,\pf)\right] e^{-\bar r(u-t)} \dr u - \cc(\alpha, \alpha^\star)e^{-\bar r(\theta-t)}\\
 \qquad+ \int_{\theta}^{+\infty} \left[D_{u}-(\alpha^\star-\bar\alpha)_{+} \ff(\delta_u,\pf)\right]\dr u \end{array} \middle|\cU_t\right],
\end{align}
where 
\begin{itemize}
    \item $\cc$ is a continuous function from~$(\RR_+)^2$ to~$\RR_+^*$,
    \item for each source of energy $\pf$, $\ff(.,\pf)$ is a derivable function from~$\RR_+$ to~$\RR_+^*$,
    \item $r, \alpha, \bar\alpha > 0$ with $\alpha > \bar\alpha$ and $\alpha^\star \sim g$ (i.e. $\alpha^\star$ is random with probability density function $g$).
\end{itemize}
Moreover, if the optimal renovation date noted $\ft$ exists in $[t,+\infty]$, we have
\begin{align*}
    \cC_{t,\delta} := R\displaystyle\EE\left[\begin{array}{l}
 \int_{t}^{\ft}\left[D_{u}-(\alpha-\bar\alpha) \ff(\delta_u,\pf)\right] e^{-\bar r(u-t)} \dr u - \cc(\alpha, \alpha^\star)e^{-\bar r(\ft-t)}\\
 \qquad+ \int_{\ft}^{+\infty} \left[D_{u}-(\alpha^\star-\bar\alpha)_{+} \ff(\delta_u,\pf)\right] \dr u \end{array} \middle|\cU_t\right].
\end{align*}
\end{definition}
At time~$t$, $D_{t}$ represents the income at $t$ when the building is "perfectly" efficient (i.e. whose the energy efficiency is less than $\bar\alpha$) while $D_{t}-(\alpha-\bar\alpha) \ff(\delta_t,\pf)$ is the income of a non-renovated building which consequently undergoes the climate transition so that its income is the difference of the income of an equivalent efficient building minus the energy costs. 

Therefore, the term~$\left[D_{u}-(\alpha-\bar\alpha) \ff(\delta_u,\pf)\right] e^{-\bar r(u-t)}$ is the actualized income at~$t$ of the property before the optimal renovation date~$\ft$. The term~$[D_{u}-(\alpha^\star-\bar\alpha)_{+} \ff(\delta_u,\pf)] e^{-\bar r(u-t)}$ is  the actualized income at~$t$ of the property after~$\ft$. This means that after the renovation, if $\alpha^\star > \bar\alpha$, the building remains inefficient and continues to incur additional energy costs and if $\alpha^\star < \bar\alpha$, it becomes efficient. Furthermore, the term~$\cc(\alpha, \alpha^\star)e^{-r(\ft-t)}$ is the actualized renovation costs which are performed at time~$\ft$. We note finally that for all~$t\geq 0$, $C_{t}-\cC_{t,\delta} \leq R\EE[\cc(\alpha, \alpha^\star)]$ i.e. the agent can always renovate the home right away, but that might not be optimal.

\paragraph{An example of the energy price function}
We can assume that the price of each type of energy~$\pf$ is a linear function (introduced in Assumption~\ref{ct-ass:housing}) of the carbon price, therefore
\begin{equation}\label{ct-eq:f}
    \ff: (\delta_t,\pf) \mapsto \ff^\pf_1 \delta_t + \ff^\pf_0\qquad t\geq 0,
\end{equation}
with $\ff^\pf_1, \ff^\pf_0>0$ and $\delta$ is the carbon price defined in the Standing Assumption~\ref{ct-sassc:price} or an example given in~\eqref{ct-eq:carbon price}.

\paragraph{An example of the renovation costs function}
We can consider that the costs of renovation of a dwelling~$\cc$, to move its energy efficiency from $x$ to $y$, is,
\begin{equation}\label{ct-eq:c}
    \cc: (x,y) \mapsto c_0 |x-y|^{1+c_1},
\end{equation}
with $c_0>0$ and $c_1\geq -1$. This choice of $\cc$ allows us to model that when a building has a poor energy efficiency, its renovation is costly.

\paragraph{An example of the distribution of the energy efficiency}

For a building whose the energy efficiency is $\alpha$, we can assume that $\alpha^\star$ follows an uniform distribution on $[0,\alpha]$. We then have
\begin{equation}\label{ct-eq:law ee}
    g = \frac{1}{\alpha}\bOne_{[0,\alpha]}.
\end{equation}

The expression~\eqref{ct-eq:def house price} can be simplified in the following proposition.
{
\begin{theorem}\label{ct-prop:housing}
    Assume that the following conditions are satisfied:
    \begin{enumerate}
        \item the carbon price function $\delta: t\mapsto \delta_t$ is non decreasing on $\RR_+$ and deterministic;
        \item the energy price $\ff(\cdot, \pf)$ is non decreasing on $\RR_+$ for all~$\pf$,
        \item and $\alpha-\bar\alpha-\EE_{\alpha^\star\sim g}\left[(\alpha^\star-\bar\alpha)_{+} \right] > 0$.
    \end{enumerate}
    Then, the market value of the building serving as the collateral to firm~$n$ at $t\geq 0$, given the carbon price sequence~$\delta$, is given by
    \begin{align}\label{ct-eq:coldynamics1}
        \cC_{t,\delta} = C_{t} - R X_{t,\delta},
    \end{align}
    where
    \begin{footnotesize}
    \begin{align}\label{ct-eq:coldynamics4}
        X_{t,\delta} := (\alpha-\bar\alpha)\int_{t}^{\ft} \ff(\delta_u,\pf) e^{-\bar r(u-t)} \dr u + \EE_{\alpha^\star\sim g}\left[\cc(\alpha, \alpha^\star)\right]e^{-\bar r(\ft-t)} +\EE_{\alpha^\star\sim g}\left[(\alpha^\star-\bar\alpha)_{+} \right]\int_{\ft}^{+\infty} \ff(\delta_u,\pf) e^{-\bar r(u-t)} \dr u,
    \end{align}
    \end{footnotesize}
    and where the optimal date of renovations~$\ft\in[t,+\infty]$ is given by 
\begin{numcases}{\ft=}
$t$ & if $\ff(\delta_{\theta},\pf) - \bar r \frac{\EE_{\alpha^\star\sim g}\left[\cc(\alpha, \alpha^\star)\right]}{\EE_{\alpha^\star\sim g}\left[\alpha-\bar\alpha-(\alpha^\star-\bar\alpha)_{+}\right]} > 0$ for all $\theta\in[t,\infty)$\label{ct-eq:optimal t_n t} \\
+\infty & if $\ff(\delta_{\theta},\pf) - \bar r \frac{\EE_{\alpha^\star\sim g}\left[\cc(\alpha, \alpha^\star)\right]}{\EE_{\alpha^\star\sim g}\left[\alpha-\bar\alpha-(\alpha^\star-\bar\alpha)_{+}\right]} < 0$ for all $\theta\in[t,\infty)$ \label{ct-eq:optimal t_n infty}\\
\theta^\star & the unique solution of $\ff(\delta_{\theta},\pf) =\bar r \frac{\EE_{\alpha^\star\sim g}\left[\cc(\alpha, \alpha^\star)\right]}{\EE_{\alpha^\star\sim g}\left[\alpha-\bar\alpha-(\alpha^\star-\bar\alpha)_{+}\right]}$ on $\theta\in[t,\infty)$ \label{ct-eq:optimal t_n}
\end{numcases} 
\end{theorem}
\begin{proof}
    Let $n\in\OneN$ and $t\geq 0$, the difference between a building with energy efficiency~$\alpha$ and an equivalent one with efficiency~$\alpha^*$ is
    \begin{align*}
        C_{t}-\cC_{t,\delta} &
        = R\EE\left[\int_{t}^{+\infty} e^{- \bar r (s-t)} D_{s}\dr s \middle|\cU_t\right]\\
        &\qquad- R\displaystyle\operatorname*{ess~sup}_{\theta\geq t} \EE\left[\begin{array}{l}
 \int_{t}^{\theta}\left[D_{u}-(\alpha-\bar\alpha) \ff(\delta_u,\pf)\right] e^{-\bar r(u-t)} \dr u - \cc(\alpha, \alpha^\star)e^{-\bar r(\theta-t)}\\
 \qquad+ \int_{\theta}^{+\infty} \left[D_{u}-(\alpha^\star-\bar\alpha)_{+} \ff(\delta_u,\pf)\right]\dr u \end{array} \middle|\cU_t\right],\\
% &=R\displaystyle\operatorname*{ess~inf}_{\theta\geq t} \EE\left[\begin{array}{l}
 %\int_{t}^{+\infty} e^{- \bar r (s-t)} D_{s}\dr s - \int_{t}^{\theta}\left[D_{s}-(\alpha-\alpha^\star) \ff(\delta_u,\pf)\right] e^{-\bar r(u-t)} \dr u \\
 %\qquad+ \cc(\alpha, \alpha^\star)e^{-\bar r(\theta-t)}- \int_{\theta}^{+\infty} e^{- \bar r (s-t)} D_{t+s}\dr s \end{array} \middle|\cU_t\right],\\
    \end{align*}
    after a few calculations, we have
    \begin{align}
         C_{t}-\cC_{t,\delta}&=R\displaystyle\operatorname*{ess~inf}_{\theta\geq t} \EE\left[ \begin{array}{l}
\int_{t}^{\theta}(\alpha-\bar\alpha) \ff(\delta_u,\pf) e^{-\bar r(u-t)} \dr u +\cc(\alpha, \alpha^\star)e^{-\bar r(\theta-t)} \\
 \qquad+ \int_{\theta}^{+\infty}(\alpha^\star-\bar\alpha)_{+} \ff(\delta_u,\pf) e^{-\bar r(u-t)} \dr u
 \end{array}\middle|\cU_t\right]\\
 &=R\displaystyle\operatorname*{ess~inf}_{\theta\geq t} \EE\left[ \EE_{\alpha^\star\sim g}\left[\begin{array}{l}
\int_{t}^{\theta}(\alpha-\bar\alpha) \ff(\delta_u,\pf) e^{-\bar r(u-t)} \dr u +\cc(\alpha, \alpha^\star)e^{-\bar r(\theta-t)} \\
 \qquad+ \int_{\theta}^{+\infty}(\alpha^\star-\bar\alpha)_{+} \ff(\delta_u,\pf) e^{-\bar r(u-t)} \dr u
 \end{array}\right]\middle|\cU_t\right]\\
 &\label{ct-eq:opti depreciation}=R\displaystyle\operatorname*{ess~inf}_{\theta\geq t} \EE\left[ \begin{array}{l}
(\alpha-\bar\alpha)\int_{t}^{\theta} \ff(\delta_u,\pf) e^{-\bar r(u-t)} \dr u + \EE_{\alpha^\star\sim g}\left[\cc(\alpha, \alpha^\star)\right]e^{-\bar r(\theta-t)} \\
 \qquad+\EE_{\alpha^\star\sim g}\left[(\alpha^\star-\bar\alpha)_{+} \right]\int_{\theta}^{+\infty} \ff(\delta_u,\pf) e^{-\bar r(u-t)} \dr u
 \end{array}\middle|\cU_t\right].
    \end{align}
    According to Standing Assumption~\ref{ct-sassc:price}, $\delta$ is deterministic. We can then write
    \begin{align}\label{ct-eq:opti depreciation 2}
        C_{t}-\cC_{t,\delta} &
        \geq R\EE\left[ \displaystyle\operatorname*{inf}_{\theta\geq t} \left\{\begin{array}{l}
(\alpha-\bar\alpha)\int_{t}^{\theta} \ff(\delta_u,\pf) e^{-\bar r(u-t)} \dr u + \EE_{\alpha^\star\sim g}\left[\cc(\alpha, \alpha^\star)\right]e^{-\bar r(\theta-t)} \\
 \qquad+\EE_{\alpha^\star\sim g}\left[(\alpha^\star-\bar\alpha)_{+} \right]\int_{\theta}^{+\infty} \ff(\delta_u,\pf) e^{-\bar r(u-t)} \dr u
 \end{array}\right\}\middle|\cU_t\right]
    \end{align}
    The function under the "inf" that we note
    \begin{small}
    \begin{align*}
        H: \theta\mapsto &(\alpha-\bar\alpha)\int_{t}^{\theta} \ff(\delta_u,\pf) e^{-\bar r(u-t)} \dr u + \EE_{\alpha^\star\sim g}\left[\cc(\alpha, \alpha^\star)\right]e^{-\bar r(\theta-t)} +\EE_{\alpha^\star\sim g}\left[(\alpha^\star-\bar\alpha)_{+} \right]\int_{\theta}^{+\infty} \ff(\delta_u,\pf) e^{-\bar r(u-t)} \dr u,
    \end{align*}
    \end{small}
    is twice differentiable on~$[t,+\infty]$, its first order derivative is 
    \begin{align*}
        H':\theta \mapsto &\left(\EE_{\alpha^\star\sim g}\left[-\bar r\cc(\alpha, \alpha^\star)+(\alpha-\bar\alpha-(\alpha^\star-\bar\alpha)_{+}) \ff(\delta_\theta,\pf)\right] \right) e^{-\bar r(\theta-t)},
    \end{align*}
    and its second order derivative is 
    \begin{align*}
        H'':\theta \mapsto &-\bar r\left[\EE_{\alpha^\star\sim g}\left[-\bar r\cc(\alpha, \alpha^\star)+(\alpha-\bar\alpha-(\alpha^\star-\bar\alpha)_{+}) \ff(\delta_\theta,\pf)\right] \right]e^{-\bar r(\theta-t)}\\& + \EE_{\alpha^\star\sim g}\left[(\alpha-\bar\alpha-(\alpha^\star-\bar\alpha)_{+})\right] \delta'_\theta\ff'(\delta_\theta,\pf) e^{-\bar r(\theta-t)}.
    \end{align*}
    Assume that $\alpha-\bar\alpha-\EE_{\alpha^\star\sim g}\left[(\alpha^\star-\bar\alpha)_{+} \right] > 0$.
   \begin{enumerate}
       \item  If a solution noted~$\theta^\star$ of~\eqref{ct-eq:find optimal t_n} following 
    \begin{align}\label{ct-eq:find optimal t_n}
    \ff(\delta_{\theta^\star},\pf) = \bar r \frac{\EE_{\alpha^\star\sim g}\left[\cc(\alpha, \alpha^\star)\right]}{\EE_{\alpha^\star\sim g}\left[\alpha-\bar\alpha-(\alpha^\star-\bar\alpha)_{+}\right] },
\end{align}
exists in $[t+\infty)$ then, by remarking that $H''(\theta^\star) = (\alpha-\alpha^\star) \delta'_\theta\ff'(\delta_{\theta^\star},\pf) e^{-\bar r(\theta^\star-t)} \geq 0$, we obtain that $\theta^\star$ is a minimum. Moreover, according to~\eqref{ct-eq:opti depreciation}, $C_{t}-\cC_{t,\delta} \leq R H(\theta^\star)$. Combining with~\eqref{ct-eq:opti depreciation 2}, we conclude that 
 \begin{align*}
     &\displaystyle\operatorname*{ess~inf}_{\theta\geq t} \EE\left[\begin{array}{l}
(\alpha-\bar\alpha)\int_{t}^{\theta} \ff(\delta_u,\pf) e^{-\bar r(u-t)} \dr u + \EE_{\alpha^\star\sim g}\left[\cc(\alpha, \alpha^\star)\right]e^{-\bar r(\theta-t)} \\
 \qquad+\EE_{\alpha^\star\sim g}\left[(\alpha^\star-\bar\alpha)_{+} \right]\int_{\theta}^{+\infty} \ff(\delta_u,\pf) e^{-\bar r(u-t)} \dr u
 \end{array}\middle|\cU_t\right] \\
     &\qquad= \EE\left[ \displaystyle\operatorname*{inf}_{\theta\geq t} \left\{\begin{array}{l}
(\alpha-\bar\alpha)\int_{t}^{\theta} \ff(\delta_u,\pf) e^{-\bar r(u-t)} \dr u + \EE_{\alpha^\star\sim g}\left[\cc(\alpha, \alpha^\star)\right]e^{-\bar r(\theta-t)} \\
 \qquad+\EE_{\alpha^\star\sim g}\left[(\alpha^\star-\bar\alpha)_{+} \right]\int_{\theta}^{+\infty} \ff(\delta_u,\pf) e^{-\bar r(u-t)} \dr u
 \end{array}\right\}\middle|\cU_t\right]\\
     &\qquad= (\alpha-\bar\alpha)\int_{t}^{\theta^\star} \ff(\delta_u,\pf) e^{-\bar r(u-t)} \dr u + \EE_{\alpha^\star\sim g}\left[\cc(\alpha, \alpha^\star)\right]e^{-\bar r(\theta^\star-t)} \\
 &\qquad\qquad+\EE_{\alpha^\star\sim g}\left[(\alpha^\star-\bar\alpha)_{+} \right]\int_{\theta^\star}^{+\infty} \ff(\delta_u,\pf) e^{-\bar r(u-t)} \dr u,
 \end{align*}
\item If \eqref{ct-eq:find optimal t_n} does not have a solution on $[t,\infty)$, then
\begin{enumerate}
    \item if for all~$\theta\in[t,\infty)$, $\ff(\delta_{\theta},\pf) - \bar r \frac{\EE_{\alpha^\star\sim g}\left[\cc(\alpha, \alpha^\star)\right]}{\EE_{\alpha^\star\sim g}\left[\alpha-\bar\alpha-(\alpha^\star-\bar\alpha)_{+}\right]} > 0$, we have $H'>0$ on $[t,\infty)$. We can then write for all~$\theta\in[t,\infty)$, $H(t) \leq H(\theta)$ i.e. the "inf" of $H$ exists and is reached in~$t$. Finally, combining to~\eqref{ct-eq:opti depreciation} implying that~$C_{t}-\cC_{t,\delta} \leq R H(t)$, we conclude that $C_{t}-\cC_{t,\delta} = R H(t)$ and~\eqref{ct-eq:optimal t_n t} follows. 
    \item if for all~$\theta\in[t,\infty)$, $\ff(\delta_{\theta},\pf) - \bar r \frac{\EE_{\alpha^\star\sim g}\left[\cc(\alpha, \alpha^\star)\right]}{\EE_{\alpha^\star\sim g}\left[\alpha-\bar\alpha-(\alpha^\star-\bar\alpha)_{+}\right]} < 0$, we have $H'<0$ on $[t,\infty)$. We can then write for all~$\theta\in[t,\infty)$, $\lim_{x\mapsto+\infty} H(x) \leq H(\theta)$ i.e. the "inf" of $H$ is reached in~$+\infty$.\\
    Assume that $(\theta_m)_{m\in\NN}$ is a non decreasing and non negative sequence so that $\lim_{m\to+\infty} \theta_m = +\infty$. According to~\eqref{ct-eq:opti depreciation}, $C_{t}-\cC_{t,\delta} \leq R H(\theta_m)$ for all $m\in\NN$. And because $H$ is continuous, we have~$C_{t}-\cC_{t,\delta} \leq R \lim_{x\mapsto+\infty} H(x)$. We conclude that that $C_{t}-\cC_{t,\delta} = R \lim_{\theta\mapsto+\infty} H(\theta)$ and~\eqref{ct-eq:optimal t_n infty} follows. 
\end{enumerate}
   \end{enumerate}
\end{proof}
}
The date of renovations~$\ft$ (obtained by~\eqref{ct-eq:optimal t_n t}, \eqref{ct-eq:optimal t_n}, and\eqref{ct-eq:optimal t_n infty}) shows that the optimal date chosen to renovate the building mainly depends on the carbon price policy~$\delta$. We also remark that the shock price due to the climate transition $\cC_{\cdot,\delta}-C_{\cdot}$ is deterministic. This is because the carbon price as well as renovation costs are deterministic. Note that $t\mapsto X_{t,\delta}$ is continuous. If, at date $t$, we realize that the optimal renovation date is $\ft$, the best thing we can do is to spend $\cc(\alpha, \alpha^\star)$ at date $t$ in order to renovate. We also have the following remarks.

\begin{remark}
In our model, the usual price of housing is (partly) offset by the costs $X_{\cdot,\delta}$ associated with the climate transition. The dwelling price could also be negative. However, we can imagine many others ways to decline the effects of transition on real estate, for example,
\begin{equation*}
\left\{
\begin{array}{ll}
 & C_t = R C_0 \exp{K_t} \\
 & K_t = \left(\dot{\chi}_t + \nu(\chi_t - K_t) \right) \dr t + \overline{\sigma} \dr\overline{B}_t - \dr X_{t,\delta},
\end{array}
\right.
\end{equation*} 
where $X$ is a jump diffusion process.
\begin{enumerate}
    \item We could assume for example that $X$ follows a homogeneous Markov process: each year $t$, the energy efficiency jumps from state $s_{t-1}$ to state $s_t$, where $s_{t-1}, s_t \in \{A, B, C, D, E, F\}$, so that the price increases or decreases. A heat sieve that is renovated, for example, would therefore see its rating improved then, its price jumps. We could calibrate "easily" the transition from historical data.

    \item We could introduce the climate transition policy by the jump term~$X$, inspired by~\cite{le2022corporate}. That climate policy is characterized, for all $t\geq 0$, by a process
$X_{t,\delta} = \sum_{i=1}^{N_t} R_i$ where the Poisson process $N_t$ has a constant arrival rate $\lambda>0$ and $(R_i)_{i\geq 1}$ is a sequence of i.i.d. random variables, independent from $B^{\cZ}$. The choice of $(R_i)_{i\geq 1}$ expresses the fact that the climate transition could affect real estate price positively (if, for example, the energy renovation work is carried out in a building), or negatively (if, for example, regulations on housing emissions are tightened). 
\end{enumerate}    
\end{remark}

\paragraph{An example of the optimal renovation time}
 With the example of the carbon price in~\eqref{ct-eq:carbon price}, the example of the energy price in~\eqref{ct-eq:f}, the example of the renovation costs in~\eqref{ct-eq:c}, and the example of the energy efficiency distribution in~\eqref{ct-eq:law ee}, the optimal renovation time, solution of~\eqref{ct-eq:optimal t_n} is given by
\begin{align}\label{ct-eq:ex optimal t_n}
    \ft = t_\circ + \frac{1}{\eta_\delta} \log{\left(\frac{\frac{2 c_0 \bar r}{2+c_1} \frac{\alpha^{2+c_1}}{\alpha^2-\bar\alpha^2} -  \ff^\pf_0}{\ff^\pf_1 P_{carbon}}\right)},
\end{align}
because $\EE_{\alpha^\star\sim g}\left[\alpha-\bar\alpha-(\alpha^\star-\bar\alpha)_{+}\right] = \frac{(\alpha-\bar\alpha)(\alpha+\bar\alpha)}{2 \alpha}$ and $\EE_{\alpha^\star\sim g}\left[\cc(\alpha, \alpha^\star)\right] = \frac{c_0}{2+c_1} \alpha^{1+c_1}$ when $g = \frac{1}{\alpha}\bOne_{[0,\alpha]}$.

We can clearly remark that the optimal renovation date depends on the climate transition policy ($P_{carbon}$ and $\eta_\delta$), on the energy prices ($\ff^\pf_0$ and $\ff^\pf_1$), on the renovation costs ($c_0$ and $c_1$), and on the energy efficiencies ($\alpha$ and $\bar\alpha$).

\section{Numerical experiments, estimation and calibration}\label{sec:estim calib}
In this section, we describe how the parameters of multisectoral model are estimated given the historical macroeconomic variables (consumption, labour, output, GHG emissions, housing prices, etc.).

\subsection{Calibration and estimation}\label{ct-sec:calibcoll_house}
We will calibrate the model parameters on a set of data ranging from time $\ft_0$ to $\ft_1$. In practice, $\ft_0=1978$ and $\ft_1=t_\circ = 2021$. From now on, we will discretize the observation interval into $M\in\NN^*$ steps $t_m = \ft_0 + \frac{\ft_1 -\ft_0}{M} m$ for $0\leq m\leq M$. We note $\Ttt^M := \{t_0, t_1,\hdots,t_{M}\}$. We will not be interested in convergence results here.

\subsubsection{Estimation of economic parameters} \label{ct-subsec:calibva}

Using macroeconomic data (output, labor, intermediary input , and the consumption), we can calculate the trajectory of productivity growth achieved. So we have $(\Theta_{t_m})_{1\leq m\leq M}$.
 We can then compute the estimations~$\widehat{\mu}$, $\widehat{\Gamma}$, $\widehat{\Sigma}$ and $\widehat{\varsigma}$,  parameters $\mu$, $\Gamma$, $\Sigma$, and~$\varsigma$ (all defined in Standing Assumption~\ref{ct-sassump:OU}).

As $\cZ$ is a centered O.-U., $\mu$ correspond to the mean. We have 
\begin{equation*}
\widehat{\mu} = \frac{1}{M} \sum_{m=1}^{M} {\Theta}_{t_m}. 
\end{equation*} 
Then, we can take $\varsigma$ so that $\Vv[Z_t] = 1$ for all $t\in\RR$ , then $\varsigma^2 = \Vv[Z_t]$, then
\begin{equation*}
    \widehat\varsigma = \sqrt{\frac{1}{M-1} \sum_{m=1}^{M} ({\Theta}_{t_m} - \widehat{\mu})^\top  ({\Theta}_{t_m} - \widehat{\mu})}.
\end{equation*}
For all~$1\leq m\leq M$, we then have $\widehat\cZ_{t_{m}} := \frac{{\Theta}_{t_m}-\widehat{\mu}}{\widehat\varsigma}$. If we discretize the first equation of~\eqref{ct-eq:VAR}, we also have ,
\begin{equation}
    \widehat\cZ_{t_{m}} = \left(\Ir - \frac{\ft_1-\ft_0}{M}\Gamma\right) \widehat\cZ_{t_{m-1}} + \cE_{t_m},\quad\text{with}\quad \cE_{t_m} = \Sigma (B^{\cZ}_{t_{m}} - B^{\cZ}_{t_{m-1}}) \sim \cN\left(0, \frac{\ft_1-\ft_0}{M} \Sigma \right) \label{ct-eq:Zt_k}.
\end{equation}
The discrete process $(\widehat\cZ_{t_{m}})_{1\leq m\leq M}$ is then a VAR process. The estimations~$\widehat\Sigma$ and $\widehat\Gamma$ of $\Sigma$ and $\Gamma$ are obtained directly, respectively.

\subsubsection{Calibration of the real estate parameters}

We assume that, historically, the carbon price did not impact the dwelling prices so that for all~$t\in\Ttt^M$, $X_{t,\delta}$ defined in~\eqref{ct-eq:coldynamics4} is zero. Moreover,  $C_{0}$ defined in~\eqref{ct-eq:coldynamics1}, the value of the collateral at $0$, is known. All that remains is to calibrate the parameters of the process $K$ defined in~\eqref{ct-eq:coldynamics2} and~\eqref{ct-eq:coldynamics3}. Let us consider a real estate index~$(REI_{t_m})_{0\leq m\leq M}$, then for each~$1\leq m\leq M$, $K_{t_m} := \log{(REI_{t_m})}$. For calibration, we proceed exactly as~\cite{fabozzi2012pricing}. Let assume that the long-term average of the real estate index~$\chi$, introduced in Assumption~\ref{ass:housing market}, is linear  as~\cite{frontczak2015modeling} do, therefore for all $t\in\RR_+$, $\chi_t = \varrho t + \vartheta$. The estimation of the parameters~$(\varrho, \vartheta)$ is realized prior to the others.
\begin{itemize}
    \item $\varrho$ and $\vartheta)$ can be estimated with a minimization procedure, possibly
    nonlinear, by 
    \begin{equation*}
        (\widehat\varrho, \widehat\vartheta) =  argmin_{(\varrho, \vartheta)} \left\{\sum_{k=0}^{M} (K_{t_m} - \varrho t_m - \vartheta)^2\right\}.
    \end{equation*}
    \item the estimation of the mean-reversion parameter~$\nu$ (introduced in~\eqref{ct-eq:coldynamics2}),
    \begin{equation*}
        \widehat{\nu} := \log{\frac{\sum_{m=1}^{M} K_{t_{m-1}}^2}{\sum_{m=1}^{M} K_{t_m} K_{t_{m-1}}}},
    \end{equation*}
    \item the estimation of the volatility parameter~$\overline{\sigma}$ (introduced in~\eqref{ct-eq:coldynamics2}),
    \begin{equation*}
        \widehat{\overline{\sigma}}^2 := \frac{1}{M}\sum_{m=1}^{M} (K_{t_m} - K_{t_{m-1}})^2,
    \end{equation*}
    \item From~\eqref{ct-eq:coldynamics2}, we have $1 <m\leq M$, the increments of $\overline{B}$ corresponds to
    \begin{equation*}
        u_{t_m}^{\overline{B}} := \frac{1}{\widehat{\overline{\sigma}}} \left((K{t_m} - K_{t_{m-1}}) - \left(\widehat{\varrho}+ \widehat{\nu}(\widehat{\varrho} t_{m-1} + \widehat\vartheta - K_{t_{m-1}}) \right) \frac{\ft_1-\ft_0}{M}
        \right) \sim \cN\left(0, \frac{\ft_1-\ft_0}{M}\right), 
    \end{equation*}
    and from $\cE_{t_m}$ defined in~\eqref{ct-eq:Zt_k},  the increments of $B^{\cZ}$ corresponds to
    \begin{equation*}
        u_{t_m}^{B^{\cZ}} := \widehat{\Sigma}^{-1} \cE_{t_m} \in\sim \cN\left(0, \frac{\ft_1-\ft_0}{M}\Ir_I\right).
    \end{equation*}
    We see from~\eqref{ct-eq:coldynamics3}, that 
    \begin{footnotesize}
    \begin{equation*}
        \widehat\rho^\top := \left[\frac{1}{M}\sum_{m=1}^{M} \left(u_{t_m}^{B^{\cZ}} - \frac{1}{M}\sum_{i=1}^{M} u_{t_i}^{B^{\cZ}}\right) \left(u_{t_m}^{\overline{B}} - \frac{1}{M}\sum_{i=1}^{M} u_{t_i}^{\overline{B}}\right) \right] \left[\frac{1}{M}\sum_{m=1}^{M} \left(u_{t_m}^{B^{\cZ}} - \frac{1}{M}\sum_{i=1}^{M} u_{t_i}^{B^{\cZ}}\right) \left(u_{t_m}^{B^{\cZ}} - \frac{1}{M}\sum_{i=1}^{M} u_{t_i}^{B^{\cZ}}\right)^\top \right]^{-1}.
    \end{equation*}
    \end{footnotesize}
    \item the other parameters~$\bar r, R, \alpha, \alpha^\star$ useful to compute~$X$ are given. Recall that examples of $\ff$ and $\cc$ are defined in~\eqref{ct-eq:f} and in~\eqref{ct-eq:c}. 
\end{itemize}

\subsection{Simulations}
In this section as well, the idea here is not to (re)demonstrate or improve convergence results.
\subsubsection{Of the productivities~$\cZ$ and $\cA$}\label{ct-sec:simul Theta and A}
Let $K~\in\NN$, for $0\leq k\leq K$, we note $u_k = t_\circ + \frac{t_\star -t_\circ}{K} k$ for $0\leq k\leq K$. We would like to simulate $\cZ_{u_k}$ and $\cA_{u_k}$. 
For $\cZ$, we adopt the Euler-Maruyana~\cite{maruyama1955continuous, kanagawa1988rate} scheme: we have $\cZ_{t_\circ}$ and
\begin{equation}\label{ct-eq:approx Z}
 \left\{
    \begin{array}{lll}
        \widehat\cZ_{u_k} &= \widehat\cZ_{u_{k-1}} -\widehat\Gamma \widehat\cZ_{u_{k-1}}  \frac{t_\star-t_\circ}{K} + \widehat\Sigma \eta_{u_k},\qquad \eta_{u_k}~\cN\left(0, \frac{t_\star-t_\circ}{K}\Ir_I\right) \quad\text{and}\quad k = 1,\hdots,K\\
        \widehat\cZ_{t} &= \widehat\cZ_{u_{k-1}},\qquad u_{k-1}\leq t \leq u_{k} \quad\text{and}\quad k = 1,\hdots,K\\
        %\widehat\cZ_{t} &= e^{-\widehat{\Gamma} t} \cZ_{t_\circ}\qquad t_{\circ}\leq t \leq u_{1}
    \end{array}
\right..
\end{equation}
Then, given that $\cA_{t_\circ}$ is known and $\cA_t =\int_{t_\circ}^{t}\left(\mu + \varsigma \cZ_u\right) \dr u  =  \cA_{t_\circ} + \mu (t-t_\circ ) + \varsigma \int_{t_\circ}^{t} \cZ_u \dr u$, we have
\begin{equation*}
    \int_{t_\circ}^{t} \widehat\cZ_u \dr u = \int_{u_{k-1}}^{t} \widehat\cZ_u \dr u + \sum_{i=1}^{k} \int_{u_{i-1}}^{u_{i}} \widehat\cZ_u \dr u = (t-u_{k-1}) \widehat\cZ_{u_{k-1}} + \frac{t_\star-t_\circ}{K}\sum_{i=1}^{k} \widehat\cZ_{u_{i-1}}
\end{equation*}
and then for each $k=1,\hdots,K$ and $u_{k-1}\leq t \leq u_{k}$,
\begin{equation}\label{ct-eq:approx A}
    \widehat\cA_t = \cA_{t_\circ} + \widehat\mu (t-t_\circ ) + \widehat\varsigma\left((t-u_{k-1}) \widehat\cZ_{u_{k-1}} + \frac{t_\star-t_\circ}{K}\sum_{i=1}^{k} \widehat\cZ_{u_{i-1}}\right).
\end{equation}
\begin{remark}
    We could also adapt the exact simulation of the multidimensional Ornstein–Uhlenbeck~$\cZ$ based on~\cite{li2019exact} or~\cite{de2020exact}.
\end{remark}

\subsection{Of the housing price}

We compute in order~\eqref{ct-eq:mean col} and ~\eqref{ct-eq:var col}. Since $\chi_t = \varrho t + \vartheta$, and $C_0$ and $R$ are known

    \begin{equation*}
        \widehat m_{t, 0} := \log{(R C_0)} + (\widehat\varrho t + \widehat\vartheta) - (\widehat\vartheta-K_0)e^{-\widehat\nu t} + \widehat{\overline{\sigma}}\widehat\rho^\top \sum_{k=0}^{L} e^{-\widehat\nu(t-\frac{k t}{L})} \eta_{u_\frac{k T}{L}},\quad \eta_{\frac{k t}{L}}\sim\cN\left(0, \frac{t}{L}\Ir_I\right), k=0,\hdots,L,
    \end{equation*}
    with $L\in\NN^*$, and
    \begin{equation*}
        \widehat v_{t, 0} :=    \frac{(\widehat{\overline{\sigma}})^2 (1-(\lVert\widehat\rho\rVert)^2)}{\widehat\nu}\left(1-e^{-2\widehat\nu t}\right).
    \end{equation*}
and $\widehat X$ is obtained by considering that from~\eqref{ct-eq:coldynamics4},
\begin{equation*}
 \widehat X_{t,\delta} = \cc(\alpha, \alpha^\star) e^{-r(\ft_n-t)} + (\alpha-\alpha^\star)\frac{(\ft-t)}{P}\sum_{p=1}^{P} \ff(\delta_{v_p},\pf) e^{-r(v_p-t)},
\end{equation*}
 and where $\gamma$, $k$, $r$, and $R$ are known, $\ft_n$ given by~\eqref{ct-eq:ex optimal t_n}, $u_l:= \frac{(t_\star-t) l}{L},l=0,\hdots,L$, and $v_p:= \frac{(\ft_n-t) p}{P},l=0,\hdots,P$.  \\
 
\noindent Calculating these three quantities enables us to run simulations with confidence intervals.

\section{Discussion}\label{sec:discussion}
In this section, we describe the data used to calibrate the different parameters, we perform some simulations, and we comment the results.
\subsection{Data}\label{result:data} As in~\cite{bouveret2023propagation}, we work on data related to the French economy.
\begin{enumerate}
    \item Due to data availability (precisely, we do not find public monthly/quaterly data for the intermediary inputs), we consider an annual frequency.
    \item\label{histo-macro-data} Annual consumption, labor, output, and intermediary inputs come from INSEE\footnote{The French National Institute of Statistics and Economic Studies} from 1978 to 2021 (see~\cite{insee2023sut} for details) and are expressed in billion euros, therefore $\ft_0=1978$, $\ft_1 = 2021$, and $M = 44$. 
    \item For the climate transition, we consider a time horizon of ten years with $t_\circ = 2021$ as starting point, a time step of one year and $t_\star = 2030$ as ending point.  In addition, we will be extending the curves to 2034 to see what happens after the transition, even though the results will be calculated and analyzed during the transition.
    \item The 38 INSEE sectors are grouped into four categories: \textit{Very High Emitting}, 
    \textit{Very Low Emitting}, 
    \textit{Low Emitting}, and
    \textit{High Emitting}, based on their carbon intensities.
    \item Metropolitan France housing price index comes from \textit{OECD data} and are from 1980 to 2021 (see~\cite{oecd2024housingpindex} for details) in \textit{Base 2015}. We renormalize in \textit{Base 2021}. We plot in Figure~\ref{ct-fig:housing_price_index} below.
    \begin{figure}[ht!]
        \centering
    \includegraphics[width=1.\textwidth]{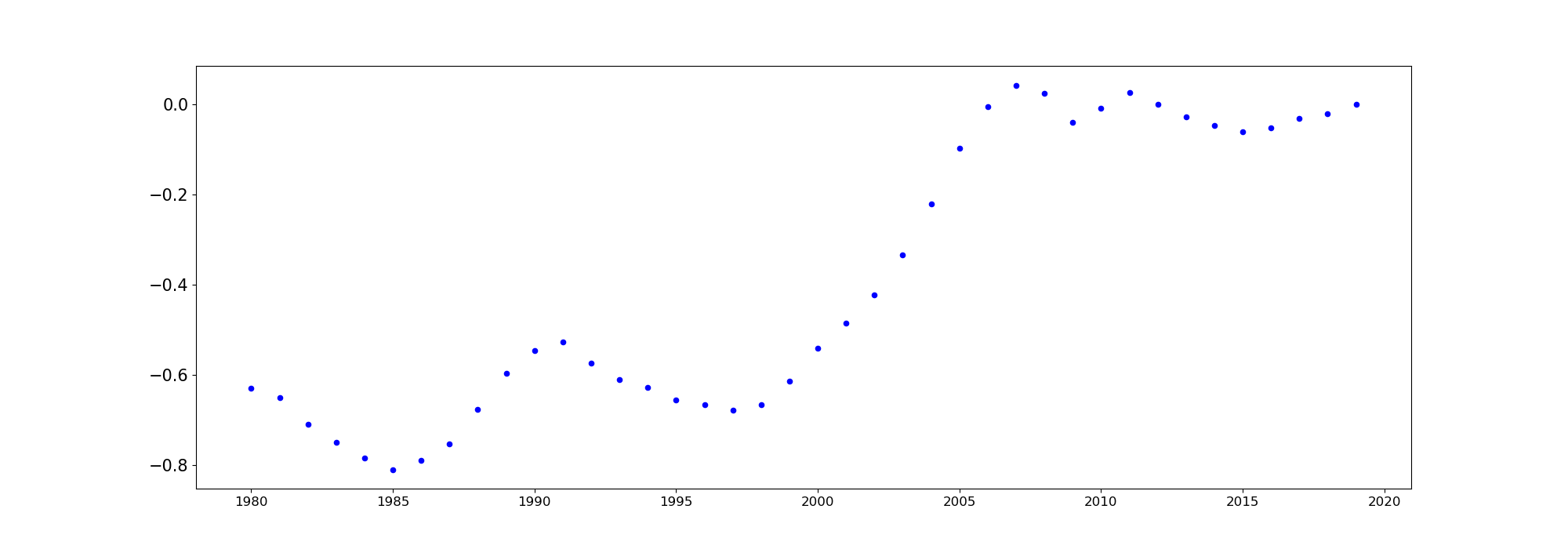}
        \caption{Log of the HPI in Base 2021 from 1980 to 2021}
        \label{ct-fig:housing_price_index}
    \end{figure}
\end{enumerate}

\subsection{Definition of the climate transition} \label{ct-subsec:calibtaxes} 
We consider four deterministic transition scenarios giving four deterministic carbon price trajectories. The scenarios used come from the NGFS simulations, whose descriptions are given by~\cite{ngfs2020scenario} as follows:
\textit{\begin{itemize}
    \item \textbf{Net Zero 2050} is an ambitious scenario that limits global warming to $1.5^\circ C$ through stringent climate policies and innovation, reaching net zero $\mathrm{CO}_2$ emissions around 2050. Some jurisdictions such as the US, EU and Japan reach net zero for all GHG by this point.
   \item \textbf{Divergent Net Zero} reaches net-zero by 2050 but with higher costs due to divergent policies introduced across sectors and a quicker phase out of fossil fuels.
   \item \textbf{Nationally Determined Contributions (NDCs)} includes all pledged policies even if not yet implemented.
   \item \textbf{Current Policies} assumes that only currently implemented policies are preserved, leading to high physical risks.
\end{itemize}}
For each scenario, we compute the carbon price $P_{carbon, 0}$ in $t_0$ and the evolution rate~$\eta_\delta$ as defined in~\eqref{ct-eq:carbon price}. 
\begin{table}[!ht]
\small \centering
\begin{tabular}{|r|r|r|r|r|}
\hline
\textit{}& \textbf{Current Policies} & \textbf{NDCs} & \textbf{Divergent Net Zero} & \textbf{Net Zero 2050} \\  \hline
\textbf{$P_{carbon, 0}$ (in euro/ton)} & 30.957&	33.321&	32.963&	34.315 \\ \hline
\textbf{$\eta_\delta$ (in \%)}& 1.693&	7.994&	12.893&	17.935\\ \hline
\end{tabular}
\caption{Carbon price parameters}
\label{ct-tab:carbon_price_params}
\end{table}
We can then compute the carbon price, whose evolution is plotted in Figure~\ref{fig:carbon_price_per_scenario}, at each date using~\eqref{ct-eq:carbon price}. 

For the energy price, we consider electricity as the unique source of energy. Then, we assume a linear relation between the electricity and the carbon price inspired by~\cite{abrell2023rising}, where a variation of the carbon price is linked withe the variation of the electricity by a the pass-through rate noted~$k$. This means that~$\ff^{elec}_1$ and~$\ff^{elec}_2$  define in~\eqref{ct-eq:f} are respectively~$k$ and $P_{elec, 0} - k \times P_{carbon, 0}$. For France, we take the electricity price $P_{elec, 0} = 0.2161$ euro per Kilowatt-hour and $k= 0.55$ (see~\cite{abrell2023rising}) ton per Kilowatt-hour. Its evolution is plotted in Figure~\ref{fig:Energy_price_per_scenario}.

For the renovation costs to improve a building for the energy efficiency~$\alpha$ to~$\alpha^\star$ as defined in~\eqref{ct-eq:c}, we take $c_0 = 0.01$ euro per kilowatt-hour and per square meter (\euro/KWh.m$^2$) and $c_1 = 0.1$.

\subsection{Estimations}
%We estimated first the economic parameters i.e. the productivity process. Then we calibrated the housing price index parameters.
\subsubsection{Economic parameters} \label{ct-subsec:result_va}
We therefore calibrate ${\mu}$, ${\varsigma}$, ${\Sigma}$, and ${\Gamma}$ as an Ornstein-Uhlenbeck, we detailed in Section~\ref{ct-subsec:calibva}.
\begin{table}[ht!]
\small \centering
\begin{tabular}{|r|r|r|r|r|}
\hline
\textit{\textbf{Emissions Level}} & \multicolumn{1}{l|}{\textbf{\begin{tabular}[c]{@{}l@{}}Very High\end{tabular}}} & \multicolumn{1}{l|}{\textbf{\begin{tabular}[c]{@{}l@{}}High \end{tabular}}} & \multicolumn{1}{l|}{\textbf{\begin{tabular}[c]{@{}l@{}}Low\end{tabular}}} & \multicolumn{1}{l|}{\textbf{\begin{tabular}[c]{@{}l@{}}Very Low\end{tabular}}}  \\ \hline
\textit{\textbf{$\times 10^{-3}$}} & 5.602&	8.475&	3.834&	12.099 \\ \hline
\end{tabular}
\caption{Parameter~${\mu}$ of the productivity}
\label{ct-tab:mu}
\end{table}

\begin{table}[!ht]
\small \centering
\begin{tabular}{|r|r|r|r|r|}
\hline
\multicolumn{1}{|r|}{\textit{\textbf{Emissions Level}}} & \multicolumn{1}{l|}{\textbf{\begin{tabular}[c]{@{}l@{}}Very High\end{tabular}}} & \multicolumn{1}{l|}{\textbf{\begin{tabular}[c]{@{}l@{}}High \end{tabular}}} & \multicolumn{1}{l|}{\textbf{\begin{tabular}[c]{@{}l@{}}Low\end{tabular}}} & \multicolumn{1}{l|}{\textbf{\begin{tabular}[c]{@{}l@{}}Very Low\end{tabular}}} \\  \hline
\textbf{Very High} & -0.201&	-0.056&	0.113&	-0.036 \\ \hline
\textbf{High} & 0.091&	0.420&	0.214&	0.015 \\ \hline
\textbf{Low} & -0.103&	-0.003&	-0.122&	0.160 \\ \hline
 \textbf{Very Low} & 0.493&	0.168&	0.290&	0.652 \\ \hline
\end{tabular}
\caption{Parameter~${\Gamma}$ of the productivity}
\label{ct-tab:gamma}
\end{table}

\noindent The eigenvalues of ${\Gamma}$ are $\{1.544, 1.057, 0.636, 0.014\}$ which implies that $-{\Gamma}$ is a Hurwitz matrix, therefore $\cZ$ is weak-stationary as assumed. Moreover, ${\varsigma} = 0.026$.
\begin{table}[!ht]
\small \centering
\begin{tabular}{|r|r|r|r|r|}
\hline
\multicolumn{1}{|r|}{\textit{\textbf{Emissions Level}}} & \multicolumn{1}{l|}{\textbf{\begin{tabular}[c]{@{}l@{}}Very High\end{tabular}}} & \multicolumn{1}{l|}{\textbf{\begin{tabular}[c]{@{}l@{}}High \end{tabular}}} & \multicolumn{1}{l|}{\textbf{\begin{tabular}[c]{@{}l@{}}Low\end{tabular}}} & \multicolumn{1}{l|}{\textbf{\begin{tabular}[c]{@{}l@{}}Very Low\end{tabular}}} \\  \hline
\textbf{Very High} & 0.473&0.013&0.072&0.092 \\ \hline
\textbf{High} & 0.013&0.208&0.039&0.037\\ \hline
\textbf{Low} & 0.072&0.039&0.059&0.020\\ \hline
 \textbf{Very Low} & 0.092&	0.037&0.020&0.068 \\ \hline
\end{tabular}
\caption{Parameter~${\Sigma}$ of the productivity}
\label{ct-tab:sigma}
\end{table}

 \subsubsection{The housing pricing index (HPI)}\label{ct-sec:result_coll}
 We write the housing price index~$K$ in \textit{Base 2021} and we apply the logarithm function. This means that $K_{t_0} = 0$ as shown in Figure~\ref{ct-fig:housing_price_index}. 
 We can therefore calibrate~$\varrho, \vartheta, \nu, \overline\sigma$, and $\rho$. The values are presented in Table~\ref{ct-tab:housing_price_index} below. 
\begin{table}[ht!]
    \centering
    \begin{tabular}{|c|c|}
\hline
\textbf{Parameter}           & \textbf{Value}  \\ \hline
\textbf{$\varrho$}           & 0.024  \\ \hline
\textbf{$\vartheta$}         & -0.884 \\ \hline
\textbf{$\nu$}               & 0.026  \\ \hline
\textbf{$\overline{\sigma}$} & 0.050  \\ \hline
\textbf{$\rho$}              &  [-0.019, -0.042, -0.017, 0.015]      \\ \hline
\end{tabular}
\caption{Housing price index parameters}
\label{ct-tab:housing_price_index}
\end{table}
\subsection{Simulations and discussions}\label{sssec:impact_on_hp}
In order to illustrate the impact of the carbon price on the housing market, we consider here 5 buildings located in the French economy whose characteristics: the price at~$t_\circ$, $C_0$, the energy efficiency~$\alpha$, the surface~$R$, are given in Table~\ref{ct-tab:impact_on_hp}

\begin{table}[ht!]
\small\centering
\begin{tabular}{|l|r|r|r|r|r|}
\hline
{ \textbf{Building }}          & \textbf{$1$} &  \textbf{$2$} &  \textbf{$3$} & \textbf{$4$}&  \textbf{$5$} \\ \hline
{ \textbf{$C_n^0$}} & 4000& 4000& 4000& 4000& 4000\\ \hline\hline
{ \textbf{$\alpha$}}          & 320.& 253.& 187.& 120.&70.\\ \hline\hline
{ \textbf{$R$}} & 25.0 &  25.0 & 25.0 & 25.0 & 25.\\ \hline
\end{tabular}
\caption{Characteristics of the building}
\label{ct-tab:impact_on_hp}
\end{table}
Precisely, we consider 5 apartments of 25 square meters whose price of the square meter fixed to 4000 euros in $t_\circ=2021$ is the same for all, but whose the energy efficiency are different.
Moreover, we assume that the optimal energy efficiency equals to $\alpha^\star = 70$ kilowatt hour per square meter per year (see~\cite{TotalEnergies2024dpe}) is reached. We can compute and summarize in table~\ref{ct-tab:opti renov date} the optimal renovation date whose the expression is given in~\eqref{ct-eq:ex optimal t_n}.
\begin{table}[ht!]
\small \centering
\begin{tabular}{|r|r|r|r|r|r|}
\hline
\textit{\textbf{Emissions level}} & \textbf{Building 1} & \textbf{Building 2} & \textbf{Building 3} &  \textbf{Building 4} & \textbf{Building 5} \\ \hline
\textit{\textbf{Current Policies}}          & 89.32&	115.16&	144.325&	185.32&	304.14 \\ \hline
\textit{\textbf{NDCs}}          & 14.21&	18.32&	22.96&	29.48&	48.38 \\ \hline
\textit{\textbf{Divergent Net Zero}}  & 8.81&	11.36&	14.24&	18.28&	29.99 \\ \hline
\textit{\textbf{Net Zero 2050}}     & 6.33&	8.17&	10.23&	13.14&	21.57 \\ \hline
\end{tabular}
\caption{Optimal renovation date (in years) per scenario and per building}
\label{ct-tab:opti renov date}
\end{table}

We observe in Table~\ref{ct-tab:opti renov date} that the optimal renovation date increases:
\begin{itemize}
    \item when the building is efficient ($\alpha$ decreases) and when $c_1 <0$: there is no point in renovating an efficient building.
    \item when the climate transition speeds up: energy costs become unbearable if we do not renovate quickly,
    \item when the renovation costs increases: if renovation costs become too high, it is better to bear the energy costs.
\end{itemize}
\begin{figure}[!ht]
    \centering
    \includegraphics[width=0.95\textwidth]{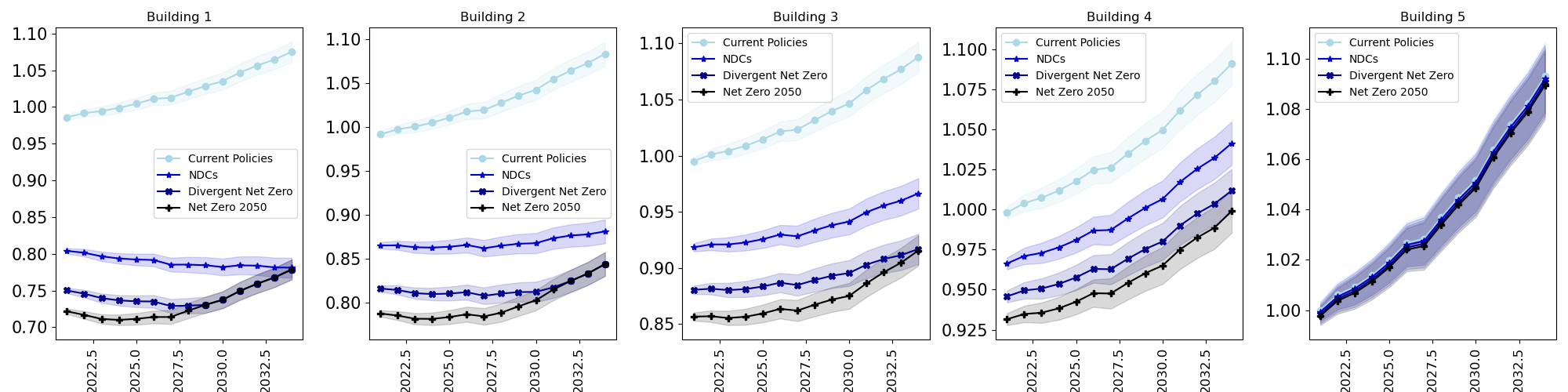}
    \caption{Apartment value per scenario and per year}
    \label{ct-fig:distrib_building}
\end{figure}
As in the case of firm values, we normalize the building values by the price of the most efficient building (Building 5) at the beginning of the transition~$t_\circ$ for ease of reading. 

We can therefore observe that, since in the scenario \textit{Current Policies}, the price of electricity does not really increase (see Figure~\ref{fig:Energy_price_per_scenario}), the optimal renovation dates are very large (much larger than the potential lifespan of the building). Therefore, if there is no climate transition, it is not necessary to renovate (for this unique reason). A direct consequence of low-cost energy and a very distant renovation date is that dwelling prices follow a persistent trend of appreciation, as observed in historical data.
\begin{table}[ht!]
\small \centering
\begin{tabular}{|r|r|r|r|r|r|}
\hline
\textit{\textbf{Emissions level}} & \textbf{Building 1} & \textbf{Building 2} & \textbf{Building 3} &  \textbf{Building 4} & \textbf{Building 5} \\ \hline
\textit{\textbf{Current Policies}}          & -1.506&	-0.870&	-0.506&	-0.209&	0.000 \\ \hline
\textit{\textbf{NDCs}}          & -22.646&	-15.526&	-9.351&	-3.848&	-0.080 \\ \hline
\textit{\textbf{Divergent Net Zero}}  &-28.007&	-20.721	&-13.480&	-6.151&	-0.179 \\ \hline
\textit{\textbf{Net Zero 2050}}     & -29.776&	-23.055&	-15.744&	-7.614&	-0.263 \\ \hline
\end{tabular}
\caption{Average annual slowdown of the housing price with respect to the \textit{Current Policies} scenario between 2021 and 2030 (in \%)}
\label{ct-tab:building_growth_slowdown}
\end{table}

For each scenario and each building, each point of the curve represents the value of the building at date $t$ if the optimal renovation date is (if $t\leq \ft$) or was (if $t>\ft$) $\ft$.
Almost all the building prices continue to grow with time as illustrated on Figure~\ref{ct-fig:distrib_building}, but these growths are more or less affected by their energy efficiency. Moreover, given that the impact of the transition on price through energy and renovation costs, the latter are stronger in the beginning. In fact, the more time passes, the closer we get to the end the climate transition ($t_\star=2030$ in our scenarios) as well as the potential date of renovations. If we look at Figure~\ref{ct-fig:housing_slowdown}, we remark that when the valuation date is later than the optimal renovation date, the best thing to do is to renovate directly. This stabilizes or even reverses the price decline. Moreover, by adding the energy costs before renovations, we could reach 20 to 30\% of depreciation when the carbon price (so the energy price) is pretty high. This seems consistent with the idea that a thermal sieve loses all its value and could become impossible to sell because of the enormous costs involved in owning it.

\begin{figure}[!ht]
    \centering
    \includegraphics[width=0.95\textwidth]{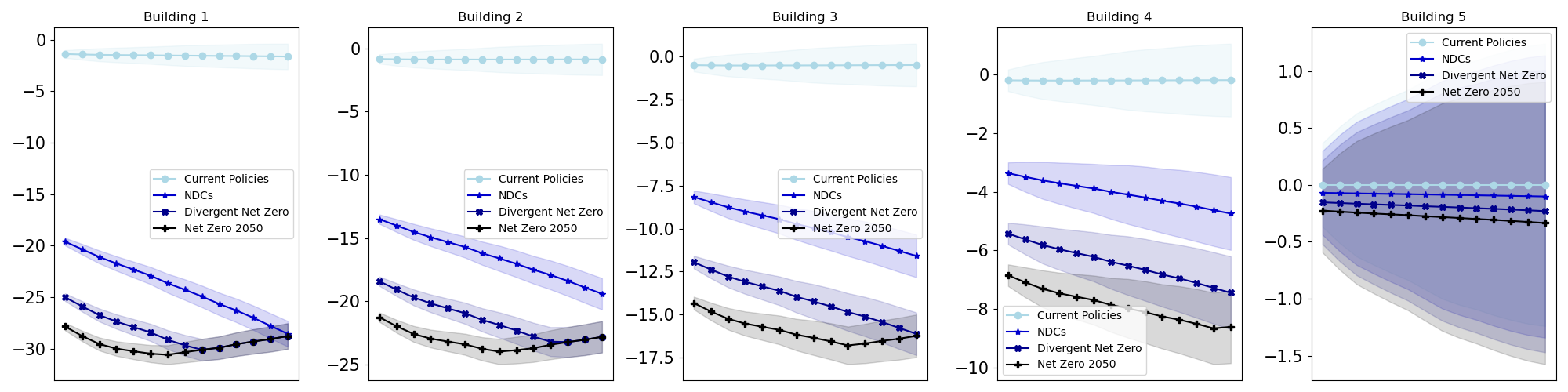}
    \caption{Apartment value slowdown per scenario and per year}
    \label{ct-fig:housing_slowdown}
\end{figure}
Finally, it should be noted that the property value discounts obtained in this work correspond to statistical observations, for example by \cite{de2016price}.

\section*{Conclusion}

The aim of this work was to propose a model to assess and quantify the impact of the climate transition on real estate prices. We first considered that real estate is an asset belonging to an economy organized in sectors, driven by dynamic and stochastic productivity following an Ornstein-Uhlenbeck dynamic, and subject to a dynamic and deterministic carbon price. Then, we assumed that the carbon price has two effects on an energy-inefficient building: it incurs additional energy costs, the price of energy is an increasing function of the carbon price, on the other hand, it is forced to be renovated to become efficient, but at a cost that depends on its inefficiency. The proposed model is built from two approaches to real estate modeling found in the literature: the \textit{income approach} and the \textit{sales comparison approach}. We obtained that the price discount due to inefficiency depends on the carbon intensity of the housing, the carbon price process, the unit renovation costs, the link between the carbon price and the energy price, etc. This work has many practical applications both in government policies (for defining of transition speed, renovation policies, etc.), in asset management (for portfolio construction notably) and in credit risk (for collateral selection and for loss modeling). Finally, it opens the door to several extensions. The building could be renovated several times over several years before reaching maximum energy efficiency. Could the renovation increase the price of the building instead of simply making it regain the historical trend? We can also model the physical risk, which would depend in particular on the location of the dwelling.

\newpage
\small
\bibliographystyle{apalike}
\bibliography{main}

\begin{thebibliography}{}

\bibitem[Abrell et~al., 2023]{abrell2023rising}
Abrell, J., Kosch, M., and Blech, K. (2023).
\newblock Rising electricity prices in europe: The impact of fuel and carbon prices.
\newblock {\em Available at SSRN 4566679}.

\bibitem[Aydin et~al., 2020]{aydin2020capitalization}
Aydin, E., Brounen, D., and Kok, N. (2020).
\newblock The capitalization of energy efficiency: Evidence from the housing market.
\newblock {\em Journal of Urban Economics}, 117:103243.

\bibitem[Bouveret et~al., 2023]{bouveret2023propagation}
Bouveret, G., Chassagneux, J.-F., Ibbou, S., Jacquier, A., and Sopgoui, L. (2023).
\newblock Propagation of carbon tax in credit portfolio through macroeconomic factors.
\newblock {\em arXiv preprint arXiv:2307.12695}.

\bibitem[de~Ayala et~al., 2016]{de2016price}
de~Ayala, A., Galarraga, I., and Spadaro, J.~V. (2016).
\newblock The price of energy efficiency in the spanish housing market.
\newblock {\em Energy Policy}, 94:16--24.

\bibitem[de~la Cruz and Jimenez, 2020]{de2020exact}
de~la Cruz, H. and Jimenez, J.~C. (2020).
\newblock Exact pathwise simulation of multi-dimensional ornstein--uhlenbeck processes.
\newblock {\em Applied Mathematics and Computation}, 366:124734.

\bibitem[{European Parliament}, 2002]{directive200291eu}
{European Parliament} (2002).
\newblock 91/ec of the european parliament and of the council of 16 december 2002 on the energy performance of buildings.
\newblock {\em Off. J. Eur. Communities}, 4(2003):L1.

\bibitem[{European Union}, 2022]{eu2022epbd}
{European Union} (2022).
\newblock Energy performance of buildings directive.

\bibitem[Eurostat, 2022]{eurostat2022}
Eurostat (2022).
\newblock Product - products datasets - eurostat.

\bibitem[Fabozzi et~al., 2012]{fabozzi2012pricing}
Fabozzi, F.~J., Shiller, R.~J., and Tunaru, R.~S. (2012).
\newblock A pricing framework for real estate derivatives.
\newblock {\em European Financial Management}, 18(5):762--789.

\bibitem[Franke and Nadler, 2019]{franke2019energy}
Franke, M. and Nadler, C. (2019).
\newblock Energy efficiency in the german residential housing market: Its influence on tenants and owners.
\newblock {\em Energy policy}, 128:879--890.

\bibitem[Frontczak and Rostek, 2015]{frontczak2015modeling}
Frontczak, R. and Rostek, S. (2015).
\newblock Modeling loss given default with stochastic collateral.
\newblock {\em Economic Modelling}, 44:162--170.

\bibitem[Gobet and She, 2016]{gobet2016perturbation}
Gobet, E. and She, Q. (2016).
\newblock {Perturbation of Ornstein-Uhlenbeck stationary distributions: expansion and simulation}.
\newblock working paper or preprint.

\bibitem[Gravelle and Rees, 2004]{gravelle2004microeconomics}
Gravelle, H. and Rees, R. (2004).
\newblock {\em Microeconomics}.
\newblock Pearson education.

\bibitem[{INSEE}, 2023]{insee2023sut}
{INSEE} (2023).
\newblock Summary tables : Sut and tiea in 2021 - the national accounts ... - insee.
\newblock \url{https://www.insee.fr/en/statistiques/6439451?sommaire=6439453}.
\newblock Accessed: Jan. 16, 2024.

\bibitem[Kanagawa, 1988]{kanagawa1988rate}
Kanagawa, S. (1988).
\newblock On the rate of convergence for maruyama’s approximate solutions of stochastic differential equations.
\newblock {\em Yokohama Math. J}, 36(1):79--86.

\bibitem[Le~Guenedal and Tankov, 2022]{le2022corporate}
Le~Guenedal, T. and Tankov, P. (2022).
\newblock Corporate debt value under transition scenario uncertainty.
\newblock {\em Available at SSRN 4152325}.

\bibitem[Li and Wu, 2019]{li2019exact}
Li, C. and Wu, L. (2019).
\newblock Exact simulation of the ornstein--uhlenbeck driven stochastic volatility model.
\newblock {\em European Journal of Operational Research}, 275(2):768--779.

\bibitem[Maruyama, 1955]{maruyama1955continuous}
Maruyama, G. (1955).
\newblock Continuous markov processes and stochastic equations.
\newblock {\em Rendiconti del Circolo Matematico di Palermo}, 4:48--90.

\bibitem[Moody's, 2022]{moodys2022}
Moody's (2022).
\newblock Non-performing and re-performing loan securitizations methodology.

\bibitem[{NGFS}, 2022]{ngfs2020scenario}
{NGFS} (2022).
\newblock {NGFS Scenarios Portal}.
\newblock \href{https://www.ngfs.net/ngfs-scenarios-portal/}{{NGFS Scenarios Portal}}.

\bibitem[OECD, 2024a]{oecd2024household}
OECD (2024a).
\newblock Oecd compendium of productivity indicators 2021 – household spending.

\bibitem[OECD, 2024b]{oecd2024realestate}
OECD (2024b).
\newblock Oecd compendium of productivity indicators 2021 – investment by asset type.

\bibitem[{OECD Stat}, 2024]{oecd2024housingpindex}
{OECD Stat} (2024).
\newblock {\em National and Regional House Price Indices: France}.
\newblock {OECD Stat}.

\bibitem[Schulz, 2003]{schulz2003valuation}
Schulz, R. (2003).
\newblock {\em Valuation of properties and economic models of real estate markets}.
\newblock PhD thesis, Humboldt University of Berlin.

\bibitem[Ter~Steege and Vogel, 2021]{ter2021german}
Ter~Steege, L. and Vogel, E. (2021).
\newblock German residential real estate valuation under ngfs climate scenarios.
\newblock Technical report, Technical Paper.

\bibitem[{Total Energies}, 2024]{TotalEnergies2024dpe}
{Total Energies} (2024).
\newblock Que signifie la classe énergie d'un logement?
\newblock Accessed: 2024-07-15.

\end{thebibliography}
\newpage

\appendix

%%%%%%%%%%%%%%%%%%%%%%%%%%%%%%%%%%%%%%%
%%%%%%%%%%%%%%%%%%%%%%%%%%%%%%%%%%%%%%%
\normalsize

\section{Ornstein-Uhlenbeck process} \label{ct-app:OU}

Let $t, h\geq 0$, from the second equation of~\eqref{ct-eq:VAR}, 
\begin{equation*}
    \cA_{t+h} = \cA_{t} + \mu h + \varsigma \int_{t}^{t+h} \cZ_s \dr s.
\end{equation*}
However, for all $t\in[t,t+h]$, from the first equation of~\eqref{ct-eq:VAR},
\begin{equation*}
    \cZ_s = e^{-\Gamma (s-t)} \cZ_t + \int_{t}^{s} e^{-\Gamma (s-u)} \Sigma \dr B_u^{\cZ},
\end{equation*}
therefore,
\begin{equation*}
    \cA_{t+h} = \cA_{t} + \mu h + \varsigma \left(\int_{t}^{t+h} e^{-\Gamma (s-t)} \dr s\right) \cZ_t + \varsigma \int_{t}^{t+h}  \left(\int_{t}^{s} e^{-\Gamma (s-u)} \Sigma \dr B_u^{\cZ}\right) \dr s.
\end{equation*}
We then have
\begin{equation*}
    \cA_{t+h} = \cA_{t} + \mu h + \varsigma\Upsilon_{h}\cZ_t + \varsigma \int_{t}^{t+h} e^{-\Gamma s} \left(\int_{t}^{s} e^{\Gamma u} \Sigma \dr B_u^{\cZ}\right) \dr s,
\end{equation*}
where $(\Upsilon_{h})_{h\geq 0}$ is defined in~\eqref{ct-eq:Upsilon}.\\

Let pose $X_s := e^{-\Gamma s}$ and $Y_s := \int_{t}^{s} e^{\Gamma u} \Sigma \dr B_u^{\cZ}$, for $s\in[t, t+h]$.
Then
\begin{equation*}
    \int_{t}^{t+h} e^{-\Gamma s} \left(\int_{t}^{s} e^{\Gamma u} \Sigma \dr B_u^{\cZ}\right) \dr s = -\Gamma^{-1}\int_{t}^{t+h} \dr X_s \cdot Y_s.
\end{equation*}
By integration by parts, as $[X, Y]_u = 0$ for all $u \geq 0$, we have
\begin{equation*}
    \begin{split}
        \int_{t}^{t+h} \dr X_s \cdot Y_s &= X_{t+h} \cdot Y_{t+h} - X_{t} \cdot Y_{t} - \int_{t}^{t+h} X_s \cdot \dr Y_s\\
        &= e^{-\Gamma (t+h)} \int_{t}^{t+h} e^{\Gamma u} \Sigma \dr B_u^{\cZ} - \int_{t}^{t+h} \Sigma \dr B_s^{\cZ}\\
        &= \int_{t}^{t+h} \left(e^{-\Gamma (t+h -s)} - \Ir_I \right) \Sigma \dr B_s^{\cZ}.
    \end{split}
\end{equation*}
Finally
\begin{equation*}
    \cA_{t+h} = \cA_{t} + \mu h + \varsigma\Upsilon_{h}\cZ_t - \varsigma\Gamma^{-1} \int_{t}^{t+h} \left(e^{-\Gamma (t+h -s)} - \Ir_I \right) \Sigma \dr B_s^{\cZ}.
\end{equation*}
The conclusion follows.

\section{Figures}

\begin{figure}[!ht]
    \centering
    \begin{subfigure}[b]{0.49\textwidth}
        \includegraphics[width=0.95\textwidth]{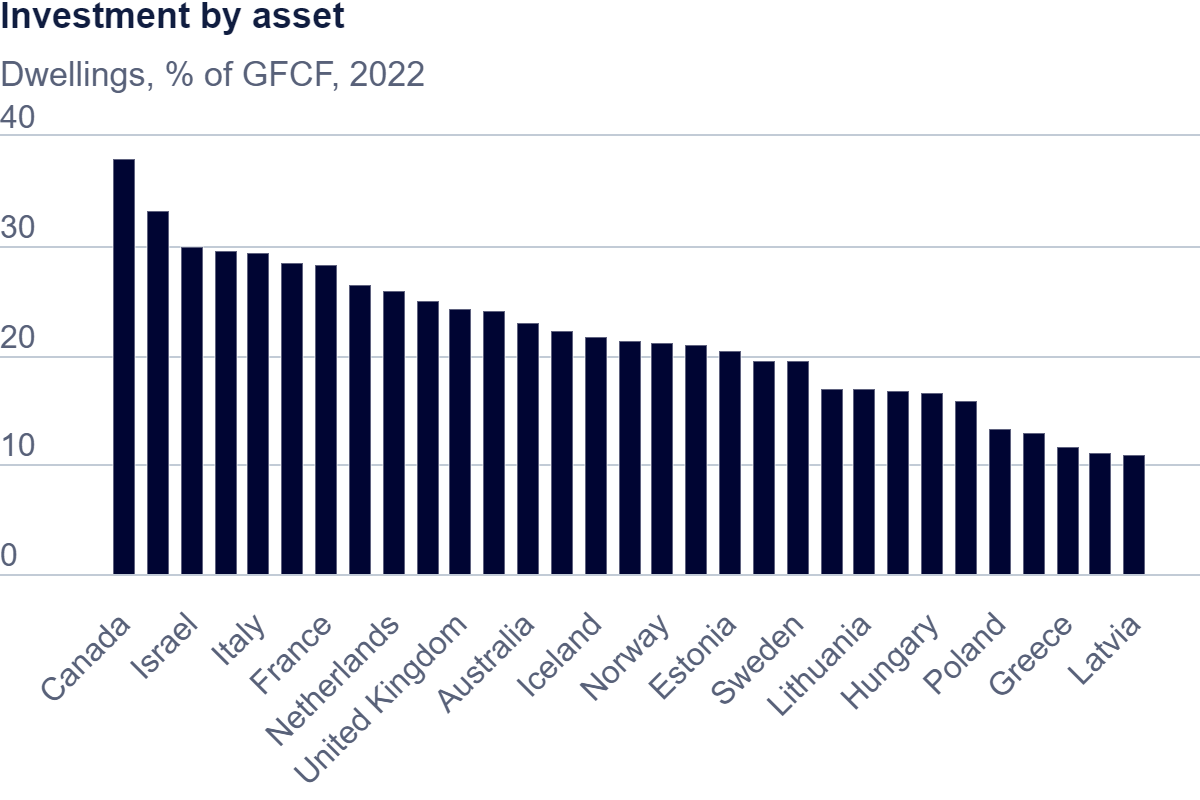}
    \caption{Share (\%) of Dwelling in Investment (\cite{oecd2024realestate})}
    \label{ct-fig:housing_share}
    \end{subfigure}
    \begin{subfigure}[b]{0.49\textwidth}
        \includegraphics[width=0.95\textwidth]{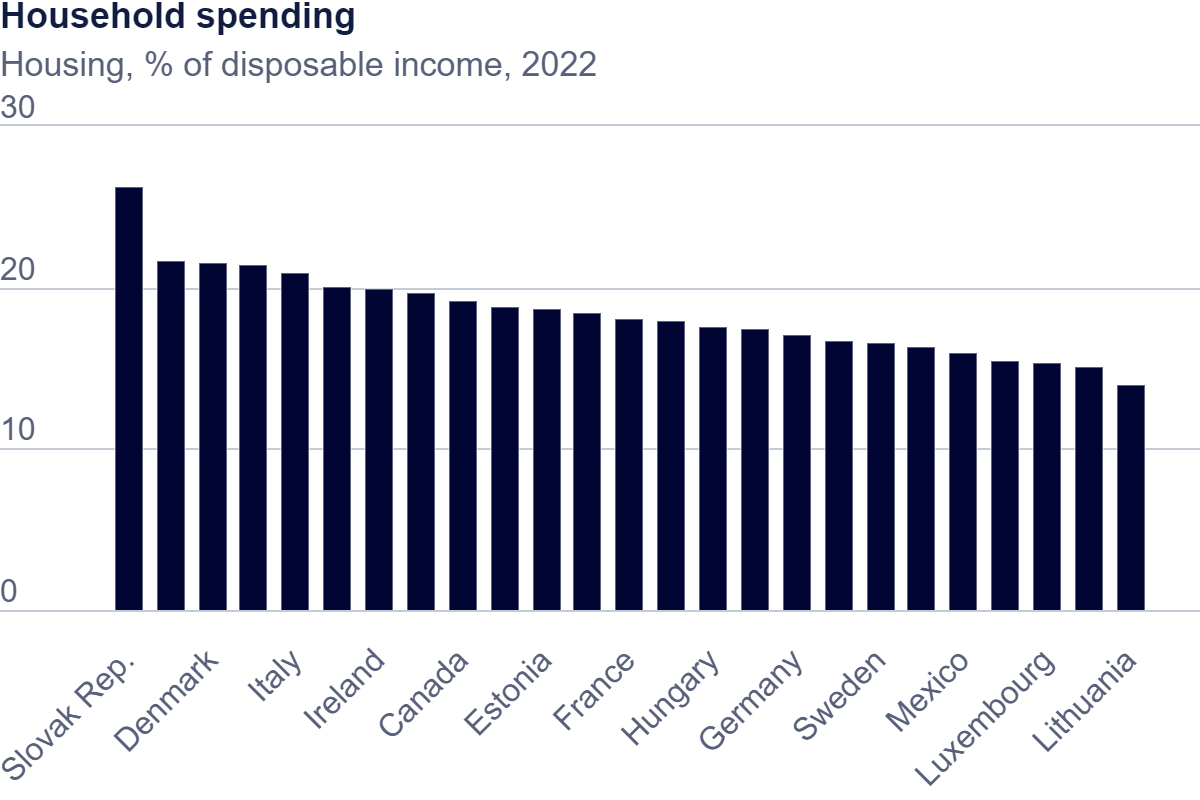}
    \caption{Share (\%) of disposable income (\cite{oecd2024household})}
    \label{ct-fig:housing_household}
    \end{subfigure}
    \caption{OECD in 2022}
    \label{fig:housing_ocde}
\end{figure}

\begin{figure}[!ht]
    \centering
    \begin{subfigure}[b]{0.49\textwidth}
        \includegraphics[width=\textwidth]{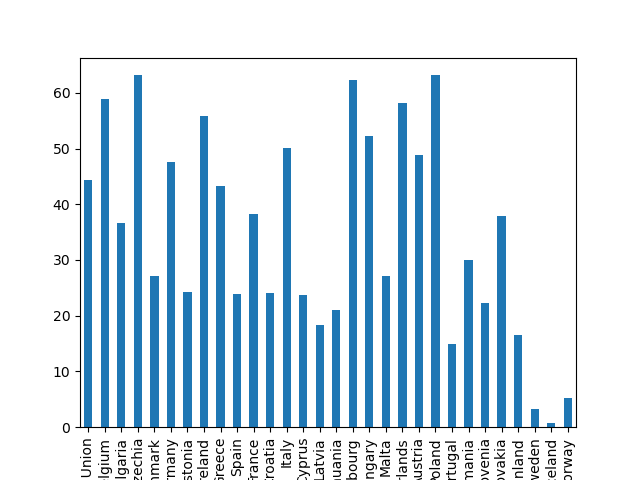}
        \caption{total households emissions}
        \label{ct-fig:part_GHG_emissions_on_households}
    \end{subfigure}
    \hfill
    \begin{subfigure}[b]{0.49\textwidth}
        \includegraphics[width=\textwidth]{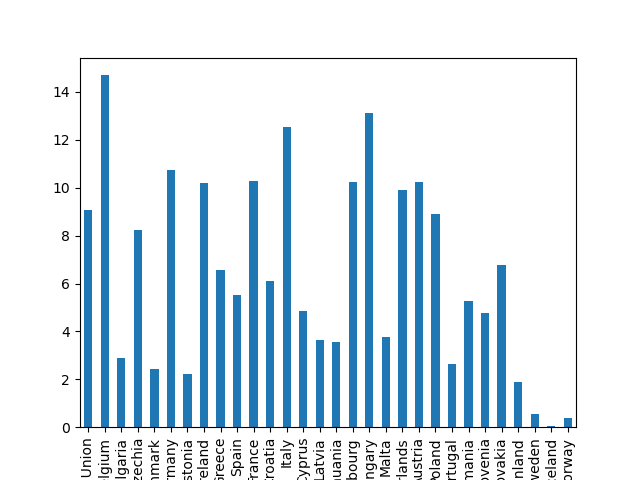}
        \caption{total economy emissions}
        \label{ct-fig:part_GHG_emissions_on_total}
    \end{subfigure}
    \caption{Part in \% of households GHG emissions on heating and cooking}
    \label{ct-fig:part_GHG_emissions}
\end{figure}

\begin{figure}[!ht]
    \centering
    \begin{subfigure}[b]{0.49\textwidth}
        \includegraphics[width=0.95\textwidth]{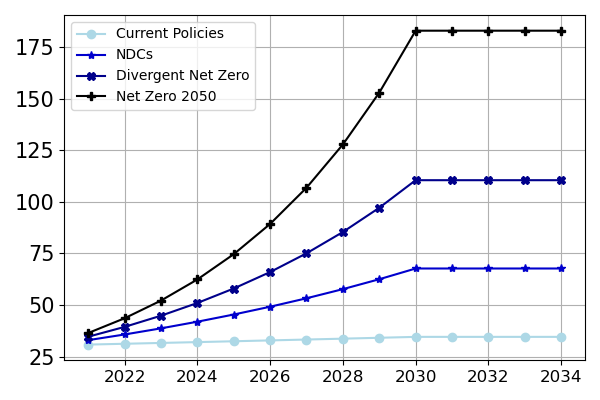}
        \caption{Carbon price}
        \label{fig:carbon_price_per_scenario}
    \end{subfigure}
    \begin{subfigure}[b]{0.49\textwidth}
        \includegraphics[width=0.95\textwidth]{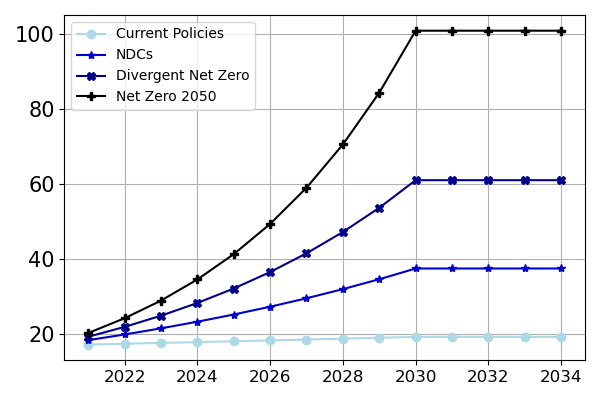}
        \caption{Energy price}
        \label{fig:Energy_price_per_scenario}
    \end{subfigure}
    \caption{per scenario and per year}
    \label{fig:cp}
\end{figure}

\end{document}